\documentclass[10pt, conference, compsocconf]{IEEEtran}
\ifCLASSINFOpdf
\usepackage[pdftex]{graphicx}
\else
\fi
\usepackage[cmex10]{amsmath}
\usepackage{algorithmic}
\usepackage{url}
\newtheorem{theorem}{Theorem}

\usepackage{adjustbox}
\usepackage{algorithm}
\usepackage{color}
\usepackage{multirow}
\usepackage{setspace}
\usepackage{lipsum}

\let\OLDthebibliography\thebibliography
\renewcommand\thebibliography[1]{
  \OLDthebibliography{#1}
  \setlength{\parskip}{4pt}
  \setlength{\itemsep}{0pt plus 0.3ex}
}

\begin{document}
\title{Significantly Improving Lossy Compression for Scientific Data Sets Based on Multidimensional Prediction and Error-Controlled Quantization}
%\title{A Novel Compression Model for Significantly Improving Error-bounded HPC Data Lossy Compression}



\author{\IEEEauthorblockN{Dingwen Tao,\IEEEauthorrefmark{1}
Sheng Di,\IEEEauthorrefmark{2}
Zizhong Chen,\IEEEauthorrefmark{1} and
Franck Cappello\IEEEauthorrefmark{2}\IEEEauthorrefmark{3}}
\IEEEauthorblockA{\IEEEauthorrefmark{1}%Department of Computer Science and Engineer\\
University of California,
Riverside, CA, USA\\ \{dtao001, chen\}@cs.ucr.edu}
\IEEEauthorblockA{\IEEEauthorrefmark{2}Argonne National Laboratory, IL, USA\\
\{sdi1, cappello\}@anl.gov}
\IEEEauthorblockA{\IEEEauthorrefmark{3}University of Illinois at Urbana-Champaign, IL, USA}}


\maketitle

\begin{abstract}
\linespread{0.94}\selectfont
Today's HPC applications are producing extremely large amounts of data, such that data storage and analysis are becoming more challenging for scientific research. In this work, we design a new error-controlled lossy compression algorithm for large-scale scientific data. Our key contribution is significantly improving the prediction hitting rate (or prediction accuracy) for each data point based on its nearby data values along multiple dimensions. We derive a series of multilayer prediction formulas and their unified formula in the context of data compression. One serious challenge is that the data prediction has to be performed based on the preceding decompressed values during the compression in order to guarantee the error bounds, which may degrade the prediction accuracy in turn. We explore the best layer for the prediction by considering the impact of compression errors on the prediction accuracy. Moreover, we propose an adaptive error-controlled quantization encoder, which can further improve the prediction hitting rate considerably. The data size can be reduced significantly after performing the variable-length encoding because of the uneven distribution produced by our quantization encoder. We evaluate the new compressor on production scientific data sets and compare it with many other state-of-the-art compressors: GZIP, FPZIP, ZFP, SZ-1.1, and ISABELA. Experiments show that our compressor is the best in class, especially with regard to compression factors (or bit-rates) and compression errors (including RMSE, NRMSE, and PSNR). Our solution is better than the second-best solution by more than a 2x increase in the compression factor and 3.8x reduction in the normalized root mean squared error on average, with reasonable error bounds and user-desired bit-rates.
%In fixed error bound term, the compression factor of our compressor is nearly 130\% higher than the second-best compressor ZFP on all tested scientific data sets on average with reasonable error bounds. Also, in fixed bit-rate term, our compressor has an increase in accuracy (or reduction in NRMSE) of 5.4x over ZFP on average with user-desired bit-rates.
\end{abstract}

% \begin{IEEEkeywords}
% Lossy Compression, Scientific Data, Compression Factor, Bit-Rates, Rate-Distortion, Performance

% \end{IEEEkeywords}

% \fi
\IEEEpeerreviewmaketitle

\section{Introduction}
One of the most challenging issues in performing scientific simulations or running large-scale parallel applications today is the vast amount of data to store in disks, to transmit on networks, or to process in postanalysis. The Hardware/Hybrid Accelerated Cosmology Code (HACC), for example, can generate 20 PB of data for a single 1-trillion-particle simulation; yet a system such as the Mira supercomputer at the Argonne Leadership Computing Facility has only 26 PB of file system storage, and a single user cannot request 75\% of the total storage capacity for a simulation. Climate research also deals with a large volume of data during simulation and postanalysis. As indicated by \cite{glecler}, nearly 2.5 PB of data were produced by the Community Earth System Model for the Coupled Model Intercomparison Project (CMIP) 5, which further introduced 170 TB of postprocessing data submitted to the Earth System Grid \cite{esg}. Estimates of the raw data requirements for the CMIP6 project exceed 10 PB \cite{baker}. 

Data compression offers an attractive solution for large-scale simulations and experiments because it enables significant reduction of data size while keeping critical information available to preserve discovery opportunities and analysis accuracy. Lossless compression preserves 100\% of the information; however, it suffers from limited compression factor (up to 2:1 in general \cite{ratana}), which is far less than the demand of large-scale scientific experiments and simulations. Therefore, only lossy compression with user-set error controls can fulfill user needs in terms of data accuracy and of large-scale execution demand.

The key challenge in designing an efficient error-controlled lossy compressor for scientific research applications is the large diversity of scientific data. Many of the existing lossy compressors (such as SZ-1.1 \cite{di} and ISABELA \cite{laksh}) try to predict the data by using curve-fitting method or spline interpolation method. The effectiveness of these compressors highly relies on the smoothness of the data in local regions. However, simulation data often exhibits fairly sharp or spiky data changes in small data regions, which may significantly lower the prediction accuracy of the compressor and eventually degrade the compression quality. NUMARCK \cite{chen} and SSEM \cite{sasaki} both adopt a quantization step in terms of the distribution of the data (or quantile), which can mitigate the dependence of smoothness of data; however, they are unable to strictly control the compression errors based on the user-set bounds. ZFP \cite{lindstrom} uses an optimized orthogonal data transform that does not strongly rely on the data smoothness either; however, it requires an exponent/fixed-point alignment step, which might not respect the user error bound when the data value range is huge (as shown later in the paper). And its optimized transform coefficients are highly dependent on the compression data and cannot be modified by users.

In this work, we propose a novel lossy compression algorithm %, namely adaptive error-controlled quantization plus variable-length encoding (AEQVE), 
that can deal with the irregular data with spiky changes effectively, will still strictly respecting user-set error bounds. Specifically, the critical contributions are threefold: 
\begin{itemize}
\item We propose a multidimensional prediction model that can significantly improve the prediction hitting rate (or prediction accuracy) for each data point based on its nearby data values in multiple dimensions, unlike previous work \cite{di} that focuses only on single-dimension prediction. Extending the single-dimension prediction to multiple dimensions is challenging. Higher-dimensional prediction requires solving more complicated surface equation system involving many more variables, which become intractable especially when the number of data points used in the prediction is relatively high. However, since the data used in the prediction must be preceding decompressed values in order to strictly control the compression errors, the prediction accuracy is degraded significantly if many data points are selected for the prediction. In this paper, not only do we derive a generic formula for the multidimensional prediction model but we also optimize the number of data points used in the prediction by an in-depth analysis with real-world data cases.
\item We design an adaptive error-controlled quantization and variable-length encoding model in order to optimize the compression quality. Such an optimization is challenging in that we need to design the adaptive solution based on very careful observation on masses of experiments and the variable-length encoding has to be tailored and reimplemented to suit variable numbers of quantization intervals. 
\item We implement the new compression algorithm, namely SZ-1.4, and release the source code under a BSD license. We comprehensively evaluate the new compression method by using multiple real-world production scientific data sets across multiple domains, such as climate simulation \cite{cesm}, X-ray scientific research \cite{aps}, and hurricane simulation \cite{hurricane}. We compare our compressor with five state-of-the-art compressors: GZIP, FPZIP, ZFP, SZ-1.1, and ISABELA. Experiments show that our compressor is the best in class, especially with regard to both compression factors (or bit-rates) and compression errors (including RMSE, NRMSE, and PSNR). On the three tested data sets, our solution is better than the second-best solution by nearly a 2x increase in the compression factor and 3.8x reduction in the normalized root mean squared error on average.
\end{itemize}

The rest of the paper is organized as follows. In Section \ref{sec: problem} we formulate the error-controlled lossy compression issue. We describe our novel compression method in Section \ref{sec: prediction} (an optimized multidimensional prediction model with best-layer analysis) and Section \ref{sec: AEQVE} (an adaptive error-controlled quantization and variable-length encoding model). In Section \ref{sec: evaluation} we evaluate the compression quality using multiple production scientific data sets. In Section \ref{sec: discuss} we discuss the use of our compressor in parallel for large-scale data sets and perform an evaluation on a supercomputer. In Section \ref{sec: relate} we discuss the related work, and in Section \ref{sec:conclude} we conclude the paper with a summary and present our future work. 
\section{Problem and Metrics Description}
\label{sec: problem}

In this paper, we focus mainly on the design and implementation of a lossy compression algorithm for scientific data sets with given error bounds in high-performance computing (HPC) applications. These applications can generate multiple snapshots that will contain many variables. Each variable has a specific data type, for example, multidimensional floating-point array and string data. Since the major type of the scientific data is floating-point, we focus our lossy compression research on how to compress multidimensional floating-point data sets within reasonable error bounds. Also, we want to achieve a better compression performance measured by the following metrics:

\begin{enumerate}
\item Pointwise compression error between original and reconstructed data sets, for example, absolute error and value-range-based relative error \footnote{Note that unlike the pointwise relative error that is compared with each data value, value-range-based relative error is compared with value range.}
\item Average compression error between original and reconstructed data sets, for example, RMSE, NRMSE, and PSNR.
\item Correlation between original and reconstructed data sets
\item Compression factor or bit-rates
\item Compression and decompression speed
%\item Autocorrelation of compression errors.
\end{enumerate}

We describe these metrics in detail below. Let us first define some necessary notations.

Let the original multidimensional floating-point data set be $X = \{x_1,x_2,..., x_N\}$, where each $x_i$ is a floating-point scalar. Let the reconstructed data set be $\tilde{X} = \{\tilde{x_1},\tilde{x_2},..., \tilde{x_N}\}$, which is recovered by the decompression process. Also, we denote the range of $X$ by $R_X$, that is, $R_X = x_{max} - x_{min}$. 

We now discuss the metrics we may use in measuring the performance of a compression method.

Metric 1: For data point $i$, let $e_{abs_i} = x_i - \tilde{x_i}$, where $e_{abs_i}$ is the \textit{absolute error}; let $e_{rel_i} = e_{abs_i} / R_X$, where $e_{rel_i}$ is the \textit {value-range-based relative error}. %For convenience, we use \textit{absolute error} to refer to absolute pointwise error and use \textit{relative error} to refer to value-range-based relative error. 
In our compression algorithm, one should set either one bound or both bounds for the absolute error and the value-range-based relative error depending on their compression accuracy requirement. The compression errors will be guaranteed within the error bounds, which can be expressed by the formula $|e_{abs_i}| < eb_{abs}$ or/and $|e_{rel_i}| < eb_{rel}$ for $1\leq i \leq N$, where $eb_{abs}$ is the absolute error bound and $eb_{rel}$ is the value-range-based relative error bound.

Metric 2: To evaluate the average error in the compression, we first use the popular root mean squared error (RMSE).
\begin{equation}
rmse  = \sqrt{\frac{1}{N}\sum_{i=1}^{N} (e_{abs_i})^2}
\end{equation}

Because of the diversity of variables, we further adopt the normalized RMSE (NRMSE).
\begin{equation}
nrmse = \frac{rmse}{R_X}
\end{equation}

The peak signal-to-noise ratio (PSNR) is another commonly used average error metric for evaluating a lossy compression method, especially in visualization. It is calculated as following.
\begin{equation} \label{eq: psnr}
psnr = 20\cdot log_{10}(\frac{R_X}{rmse})
\end{equation}
PSNR measures the size of the RMSE relative to the peak size of the signal. Logically, a lower value of RMSE/NRMSE means less error, but a higher value of PSNR represents less error.

Metric 3: To evaluate the correlation between original and reconstructed data sets, we adopt the Pearson correlation coefficient $\rho$,	
\begin{equation}
\rho = \frac{cov(X,\tilde{X})}{\sigma_{X}\sigma_{\tilde{X}}},	
\end{equation}
where $cov(X,\tilde{X})$ is the covariance. This coefficient is a measurement of the linear dependence between two variables, giving $\rho$ between $+1$ and $-1$, where $\rho = 1$ is the total positive linear correlation. The APAX profiler \cite{wegener} suggests that the correlation coefficient between original and reconstructed data should be 0.99999 (``five nines") or better. 

Metric 4: To evaluate the size reduce as a result of the compression, we use the compression factor $CF$,
\begin{equation}
CF(F) = \frac{filesize(F_{orig})}{filesize(F_{comp})},
\end{equation}
or the bit-rate (bits/value),
\begin{equation}
BR(F) = \frac{filesize_{bit}(F_{comp})}{N},
\end{equation}
where $filesize_{bit}$ is the file size in bits and $N$ is the data size. The bit-rate represents the amortized storage cost of each value. For a single/double floating-point data set, the bit-rate is 32/64 bits per value before a compression, while the bit-rate will be less than 32/64 bits per value after a compression. Also, $CF$ and $BR$ have a mathematical relationship as $BR(F) * CF(F) = 32/64$ so that a lower bit-rate means a higher compression factor.
 
Metric 5: To evaluate the speed of compression, we compare the throughput  (bytes per second) based on the execution time of both compression and decompression with other compressors.
\section{Prediction Model Based on Mutidimensional Scientific Data Sets}
\label{sec: prediction}

In Sections \ref{sec: prediction} and \ref{sec: AEQVE}, we present our novel compression algorithm. At a high level, the compression process involves three steps: (1) predict every data value through our proposed multilayer prediction model; (2) adopt an error-controlled quantization encoder with an adaptive number of intervals; and (3) perform a variable-length encoding technique based on the uneven distributed quantization codes. In this section, we first present our new multilayer prediction model designed for multidimensional scientific data sets. Then, we give a solution for choosing the best layer for our multilayer prediction model. We illustrate how our prediction model works using two-dimensional data sets as an example. 

\subsection{Prediction Model for Multidimensional Scientific Data Sets}
Consider a two-dimensional data set on a uniform grid of size $M \times N$, where $M$ is the size of second dimension and $N$ is the size of first dimension. We give each data point a global coordinate $(i,j)$, where $0 < i \leq M$ and $0 < j \leq N$.

In our compression algorithm, we process the data point by point from the low dimension to the high dimension. Assume that the coordinates of the current processing data point are $(i_0,j_0)$ and the processed data points are $(i,j)$, where $i < i_0$ or $i = i_0$, $j < j_0$, as shown in Figure \ref{fig: f1}. The figure also shows our definition of ``layer" around the processing data point $(i_0,j_0)$.  We denote the data subset $S_{i_0 j_0}^n$ and $T_{i_0 j_0}^n$ by	
\begin{align*}
S_{i_0 j_0}^n & = \{(i_0-k_1,j_0-k_2) | 0 \leq k_1,k_2 \leq n\} \setminus \{(i_0,j_0)\} \\
T_{i_0 j_0}^n & = \{(i_0-k_1,j_0-k_2) | 0 \leq k_1+k_2 \leq 2n-1, k_1,k_2 \geq 0\}.
\end{align*}
Since the data subset $S_{i_0 j_0}^n$ contains the layer from the first one to the \textit{n}th one, we call $S_{i_0 j_0}^n$ ``\textit{n-layer data subset}."

\begin{figure}[t]
\centering
\includegraphics[scale=0.33]{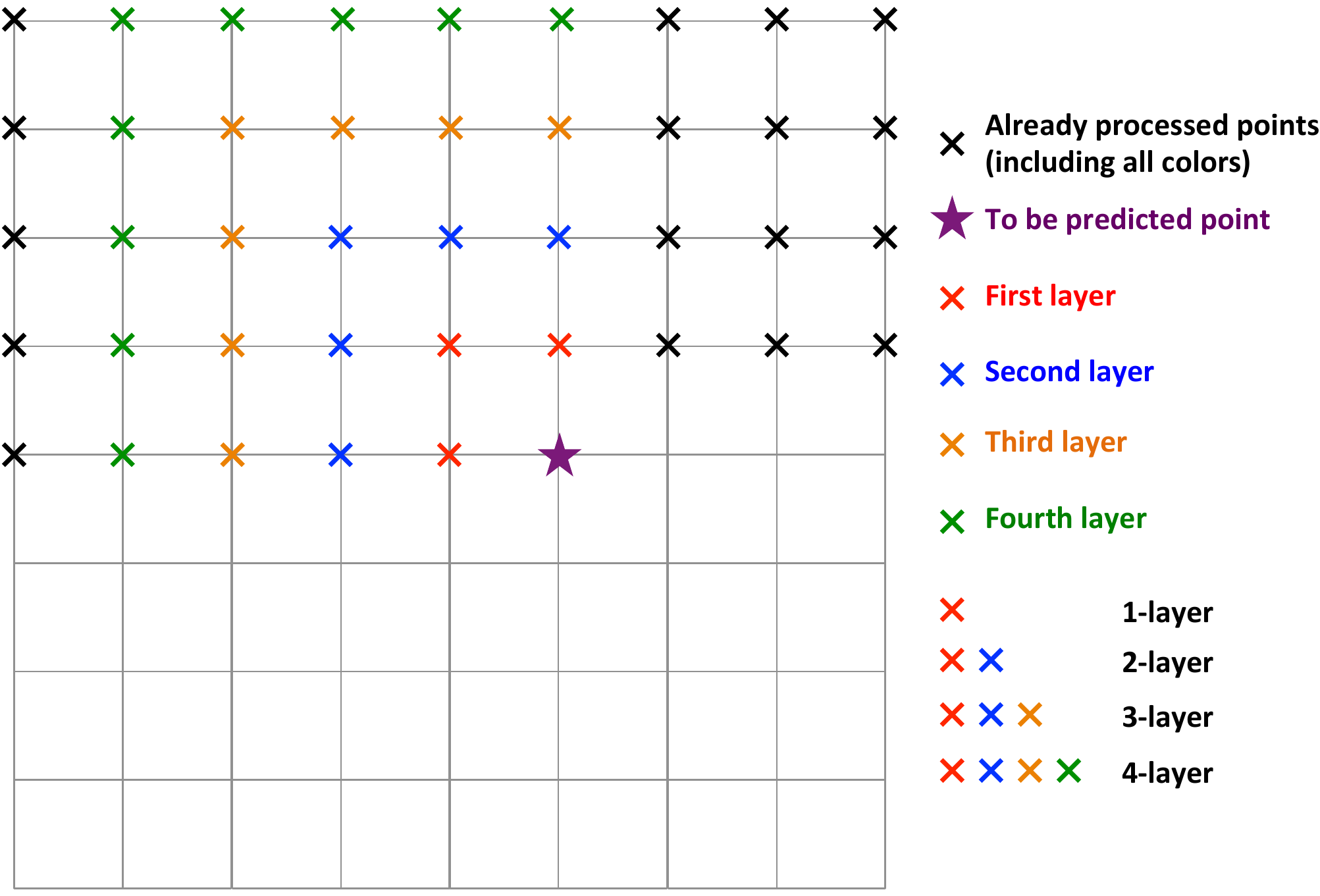}
\caption{Example of $9\times9$ two-dimensional data set showing the processed / processing data and the data in different layers of the prediction model.}
\label{fig: f1}
\vspace{-4mm}
\end{figure}

Now we build a prediction model for two-dimensional data sets using the $n(n+2)$ symmetric processed data points in the \textit{n}-layer data subset $S_{i_0 j_0}^n$ to predict data $(i_0,j_0)$.

First, let us define a three-dimensional surface, called the ``\textit{prediction surface}," with the maximum order of $2n-1$ as follows. 
\vspace{-1mm}
\begin{small}
\begin{equation} \label{eq: surface}
f(x,y) = \sum\limits_{0 \leq i+j \leq 2n-1}^{i,j \geq 0} a_{i,j}x^i y^j
\end{equation}
\end{small}

The surface $f(x,y)$ has $n(2n+1)$ coefficients, so we can construct a linear system with $n(2n+1)$ equations by using the coordinates and values of $n(2n+1)$ data points. And then solve this system for these $n(2n+1)$ coefficients; consequently, we build the prediction surface $f(x,y)$. However, the problem is that not every linear system has a solution, which also means not every set of $n(2n+1)$ data is able to be on the surface at the same time. Fortunately, we demonstrate that the linear system constructed by the $n(2n+1)$ data in $T_{i_0 j_0}^n$ can be solved with an explicit solution. Also, we demonstrate that $f(i_0,j_0)$ can be expressed by the linear combination of the data values in $S_{i_0 j_0}^n$.

Now let us give the following theorem and proof.

\begin{theorem}
The $n(2n+1)$ data in $T_{i_0 j_0}^n$ will determine a surface $f(x,y)$ shown in equation (\ref{eq: surface}), and the value of $f(i_0,j_0)$ equals $\sum\limits_{0 \leq k_1,k_2 \leq n}^{(k_1,k_2) \neq (0,0)} (-1)^{k_1+k_2+1}\binom{n}{k_1} \binom{n}{k_2} V(i_0-k_1,j_0-k_2)$, where $\binom{n}{k}$ is the binomial coefficient and $V(i,j)$ is the data value of $(i,j)$ in $S_{i_0 j_0}^n$.
\end{theorem}

\begin{proof}
We transform the coordinate of each data point in $T_{i_0 j_0}^n$ to a new coordinate as $(i_0-k_1, j_0-k_2) \rightarrow (k_1, k_2)$.

Then, using their new coordinates and data values, we can construct a linear system with $n(2n+1)$ equations as
\begin{small}
\begin{equation} \label{eq: system}
V(k_1,k_2) = \sum\limits_{0 \leq i+j \leq 2n-1}^{i,j \geq 0} a_{i,j}k_1^i k_2^j,
\end{equation}
\end{small}
where $0 \leq k_1+k_2 \leq 2n-1, k_1,k_2 \geq 0$.

Let us denote $F$ as follows.
\begin{small}
\begin{equation} \label{eq: F}
F = \sum\limits_{0 \leq k_1,k_2 \leq n}^{(k_1,k_2) \neq (0,0)} (-1)^{k_1+k_2+1} \binom{n}{k_1} \binom{n}{k_2} V(k_1,k_2)
\end{equation}
\end{small}

For any coefficient $a_{l,m}$, $\sum\limits_{0 \leq i+j \leq 2n-1}^{i,j \geq 0} a_{i,j}k_1^i k_2^j$ only has one term containing $a_{l,m}$, which is $k_1^l k_2^m \cdot a_{l,m}$.

Also, from equations (\ref{eq: system}) and (\ref{eq: F}), $F$ contains
$(\sum\limits_{0 \leq k_1,k_2 \leq n}^{(k_1,k_2) \neq (0,0)} (-1)^{k_1+k_2+1} \binom{n}{k_1} \binom{n}{k_2} k_1^l k_2^m) \cdot a_{l,m}$.

And because
\begin{small}
\begin{align*}
&\sum\limits_{0 \leq k_1,k_2 \leq n}^{(k_1,k_2) \neq (0,0)} (-1)^{k_1+k_2+1} \binom{n}{k_1} \binom{n}{k_2} k_1^l k_2^m \\
&=  \sum\limits_{0 \leq k_1,k_2 \leq n} (-1)^{k_1+k_2+1} \binom{n}{k_1} \binom{n}{k_2} k_1^l k_2^m + 0^{l+m} \\
&= -\sum\limits_{0 \leq k_1,k_2 \leq n} (-1)^{k_1+k_2} \binom{n}{k_1} \binom{n}{k_2} k_1^l k_2^m + 0^{l+m} \\
&= -\sum\limits_{0 \leq k_1 \leq n} (-1)^{k_1} \binom{n}{k_1} k_1^l \cdot \sum\limits_{0 \leq k_2 \leq n} (-1)^{k_2} \binom{n}{k_2} k_2^m + 0^{l+m}.
\end{align*}
\end{small}

For $l+m \leq 2n+1$, either $l$ or $m$ is smaller than $n$. Also, from the theory of finite differences \cite{brenner}, $\sum\limits_{0 \leq i \leq n} (-1)^{i} \binom{n}{i} P(x) = 0$ for any polynomial $P(x)$ of degree less than $n$, so either $\sum\limits_{0 \leq k_1 \leq n} (-1)^{k_1} \binom{n}{k_1} k_1^l = 0$ or $\sum\limits_{0 \leq k_2 \leq n} (-1)^{k_2} \binom{n}{k_2} k_2^m = 0$.

Therefore, $F$ contains $0^{l+m} \cdot a_{l,m}$, so \\ $F = \sum\limits_{0 \leq l+m \leq 2n-1}^{l,m \geq 0} 0^{l+m} \cdot a_{l,m} = a_{0,0}$ and \\
\resizebox{\hsize}{!}{$f(0,0) = a_{0,0} = \sum\limits_{0 \leq k_1,k_2 \leq n}^{(k_1,k_2) \neq (0,0)} (-1)^{k_1+k_2+1} \binom{n}{k_1} \binom{n}{k_2} V(k_1,k_2)$}.

We transform the current coordinate to the previous one reversely, namely, $(k_1, k_2) \rightarrow (i_0-k_1, j_0-k_2)$. Thus, \\
\resizebox{\hsize}{!}{$f(i_0,j_0) = \sum\limits_{0 \leq k_1,k_2 \leq n}^{(k_1,k_2) \neq (0,0)} (-1)^{k_1+k_2+1} \binom{n}{k_1} \binom{n}{k_2} V(i_0-k_1,j_0-k_2)$}.
\end{proof}

From this theorem, we know that the value of  $(i_0, j_0)$ on the prediction surface, $f(i_0, j_0)$, can be expressed by the linear combination of the data values in $S_{i_0 j_0}^n$. Hence, we can use the value of $f(i_0, j_0)$ as our predicted value for $V(i_0, j_0)$. In other words, we build our prediction model using the data values in $S_{i_0 j_0}^n$ as follows.
\begin{equation}
\resizebox{\hsize}{!}{
$f(i_0,j_0) = \sum\limits_{0 \leq k_1,k_2 \leq n}^{(k_1,k_2) \neq (0,0)} (-1)^{k_1+k_2+1} \binom{n}{k_1} \binom{n}{k_2} V(i_0-k_1,j_0-k_2)$
}
\end{equation}

We call this prediction model using \textit{n}-layer data subset $S_{i_0 j_0}^n$ the ``\textit {n-layer prediction model}," consequently, our proposed model can be called a \textbf {multilayer prediction model}.

Also, we can derive a generic formula of the multilayer prediction model for any dimensional data sets. Because of space limitations, we give the formula as follows,
\begin{small}
\begin{equation} \label{eq: predmodel}
\begin{aligned}
f(x_1, \cdots , x_d) = & \sum\limits_{0 \leq k_1,\cdots, k_d \leq n}^{(k_1,\cdots, k_d) \neq (0,\cdots,0)} -\prod\limits_{j=1}^{d}(-1)^{k_j}\binom{n}{k_j} \\
& \cdot V(x_1-k_1,\cdots, x_d-k_d),
\end{aligned}
\end{equation}
\end{small}
where $d$ is the dimensional size of the data set and $n$ represents the ``\textit{n}-layer" used in the prediction model. Note that Lerenzo predictor \cite{lerenzo} is a special case of our multi-dimensional prediction model when $n = 1$.

\subsection{In-Depth Analysis of the Best Layer for Multilayer Prediction Model}
In Subsection III-A, we developed a general prediction model for multidimensional data sets. %by using nearby multi-layer data values. 
Based on this model, we need to answer another critical question: How many layers should we use for the prediction model during the compression process? In other words, we want to find the best $n$ for equation (\ref{eq: predmodel}).

Why does there have to exist a best $n$? We will use two-dimensional data sets to explain. We know that a better $n$ can result in a more accurate data prediction, and a more accurate prediction will bring us a better compression performance, including improvements in compression factor, compression error, and compression/decompression speed. On the one hand, a more accurate prediction can be achieved by increasing the number of layers, which will bring more useful information along multiple dimensions. On the other hand, we also note that data from further distance will bring more uncorrelated information (noise) into the prediction, which means that too many layers will %also bring unnecessary noises into 
degrade the accuracy of our prediction. Therefore, we infer that there has to exist a best number of layers for our prediction model.

How can we get the best $n$ for our multilayer prediction model? 

  For a two-dimensional data set, we first need to get prediction formulas for different layers by substituting 1, 2, 3, and so forth into the generic formula of our prediction model (as shown in equation (\ref{eq: predmodel})). The formulas are shown in Table \ref{tab: predformula}.

\begin{table}[t]
\centering
\caption{Formulas of 1, 2, 3, 4-layer prediction for two-dimensional data sets}
\vspace{-2mm}
\label{tab: predformula}
\begin{adjustbox}{max width=0.48\textwidth}
\begin{tabular}{|c|l|}
\hline
        & \multicolumn{1}{c|}{Prediction Formula}                            \\ \hline
1-Layer & $f(i_0,j_0) = V(i_0,j_0-1)+V(i_0-1,j_0)-V(i_0-1,j_0-1)$ \\ \hline
2-Layer & \begin{tabular}[c]{@{}l@{}}$f(i_0,j_0) = 2V(i_0-1,j_0)+2V(i_0,j_0-1)$\\ $-4V(i_0-1,j_0-1)-V(i_0-2,j_0)-V(i_0,j_0-2)$ \\ $+2V(i_0-2,j_0-1)+2V(i_0-1,j_0-2)-V(i_0-2,j_0-2)$\end{tabular}           \\ \hline
3-Layer &   \begin{tabular}[c]{@{}l@{}}$f(i_0,j_0) = 3V(i_0-1,j_0)+3V(i_0,j_0-1)$\\ $-9V(i_0-1,j_0-1)-3V(i_0-2,j_0)-3V(i_0,j_0-2)$ \\ $+9V(i_0-2,j_0-1)+9V(i_0-1,j_0-2)-9V(i_0-2,j_0-2)$ \\ $+V(i_0-3,j_0)+V(i_0,j_0-3)$ \\ $-3V(i_0-3,j_0-1)-3V(i_0-1,j_0-3)$ \\  $+3V(i_0-3,j_0-2)+3V(i_0-2,j_0-3)-V(i_0-3,j_0-3)$ \end{tabular}                                                      \\ \hline
4-Layer &  \begin{tabular}[c]{@{}l@{}}$f(i_0,j_0) = 4V(i_0-1,j_0)+4V(i_0,j_0-1)$\\ $-16V(i_0-1,j_0-1)-6V(i_0-2,j_0)-6V(i_0,j_0-2)$ \\ $+24V(i_0-2,j_0-1)+24V(i_0-1,j_0-2)$ \\ $-36V(i_0-2,j_0-2)+4V(i_0-3,j_0)+4V(i_0,j_0-3)$ \\ $-16V(i_0-3,j_0-1)-16V(i_0-1,j_0-3)+24V(i_0-3,j_0-2)$ \\  $+24V(i_0-2,j_0-3)-16V(i_0-3,j_0-3)$ \\ $-V(i_0-4,j_0)-V(i_0,j_0-4)+4V(i_0-4,j_0-1)$ \\ $+4V(i_0-1,j_0-4)-6V(i_0-4,j_0-2)-6V(i_0-2,j_0-4)$ \\ $+4V(i_0-4,j_0-3)+4V(i_0-3,j_0-4) - V(i_0-4,j_0-4)$ \end{tabular}                                                        \\ \hline
\end{tabular}
\end{adjustbox}
\vspace{-4mm}
\label {tab:predformula}
\end{table}

Then we introduce a term called the ``\textit{prediction hitting rate}," which is the proportion of the predictable data in the whole data set. We define a data point as ``\textit{predictable data}" if the difference between its original value and predicted value is not larger than the error bound. We denote the prediction hitting rate by $R_{PH} = \frac{N_{PH}}{N}$, where $N_{PH}$ is the number of predictable data points and $N$ is the size of the data set. 

In the climate simulation ATM data sets example, the hitting rates are calculated in Table \ref {tab: bestlayerATM}, based on the prediction methods described above. Here the second column shows the prediction hitting rate by using the \textbf {original} data values, denoted by
$R_{PH}^{orig}$. In this case, 2-layer prediction will be more accurate than other layers if performing the prediction on the original data values. However, in order to guarantee that the compression error (absolute or value-range-based relative) falls into the user-set error bounds, the compression algorithm must use the preceding decompressed data values instead of the original data values. Therefore, the last column of Table \ref {tab: bestlayerATM} shows the hitting rate of the prediction by using preceding \textbf {decompressed} data values, denoted by $R_{PH}^{decomp}$. In this case, 1-layer prediction will become the best one for the compression algorithm on ATM data sets.

\begin{table}[t]
\centering
\caption{Prediction hitting rate using different layers for the prediction model based on original and decompressed data values on ATM data sets}
\vspace{-2mm}
\label{tab: bestlayerATM}
\begin{tabular}{|l|c|c|}
\hline
        & $R_{PH}^{orig}$ & $R_{PH}^{decomp}$ \\ \hline
1-Layer & 21.5\%       & \textbf {19.2}\%         \\ \hline
2-Layer & \textbf {37.5}\%       & 6.5\%         \\ \hline
3-Layer & 25.8\%       & 9.8\%         \\ \hline
4-Layer & 14.5\%       & 5.9\%         \\ \hline
\end{tabular}
\vspace{-1.0\baselineskip}
\end{table}

Since the best layer \textit{n} is data-dependent, different scientific data sets may have different best layers. Thus, we give users an option to set the value of layers in the compression process. The default value in our compressor is $n=1$. %Based on the Table \ref{tab: bestlayer}, we set the default value in our compressor as $n = 1$ for two-dimensional data sets and $n = $ for three-dimensional data set.

%\begin{table}[]
%\centering
%\caption{Best Layer for production data sets in different 2D and 3D HPC applications}
%\label{tab: bestlayer}
%\begin{tabular}{|l|l|l|l|l|}
%\hline
%\textbf{2D Application} & \textbf{Best Layer} &  & \textbf{3D Application} & \textbf{Best Layer} \\ \hline
%ATM                     &                     &  & Hurricane               &                     \\ \hline
%CICE                    &                     &  &                         &                     \\ \hline
%                        &                     &  &                         &                     \\ \hline
%                        &                     &  &                         &                     \\ \hline
%\end{tabular}
%\end{table}

\section{AEQVE: Adaptive Error-controlled Quantization and Variable-length Encoding} \label{sec: AEQVE}

In this section, we present our adaptive error-controlled quantization and variable-length encoding model, namely, AEQVE, which can further optimize the compression quality. First, we introduce our quantization method, which is completely different from the traditional one. Second, using the same logic from Subsection III-B, we develop an adaptive solution to optimize the number of intervals in the error-controlled quantization. Third, we show the fairly uneven distribution produced by our quantization encoder. Finally, we reduce the data size significantly by using the variable-length encoding technique on the quantization codes.

\subsection{Error-Controlled Quantization}
The design of our error-controlled quantization is shown in Figure \ref{fig: f2}. First, we calculate the predicted value by using the multilayer prediction model proposed in the preceding section. We call this predicted value the ``\textit {first-phase predicted value}," represented by the red dot in Fig. \ref{fig: f2}. Then, we expand $2^m-2$ values from the first-phase predicted value by %increasing or reducing the even-fold of
scaling the error bound linearly; we call these values ``\textit {second-phase predicted values}," represented by the orange dots in Fig. \ref{fig: f2}. The distance between any two adjacent predicted values equals twice the error bound. Note that each predicted value will also be expanded one more error bound in both directions to form an interval with the length of twice the error bound. This will ensure that all the intervals are not overlapped. 

If the real value of the data point falls into a certain interval, we mark it as predictable data and use its corresponding predicted value from the same interval to represent the real value in the compression. In this case, the difference between the real value and predicted value is always lower than the error bound. However, if the real value doesn't fall into any interval, we mark the data point as unpredictable data. Since there are $2^m-1$ intervals, we use $2^m-1$ codes to encode these $2^m-1$ intervals. Since all the predictable data can be encoded as the code of its corresponding interval and since all the unpredictable data will be encoded as another code, we need $m$ bits to encode all $2^m$ codes. For example, we use the codes of $1, \cdots, 2^{m-1}, \cdots, 2^{m}-1$ to encode predictable data and use the code of $0$ to encode unpredictable data. This process is quantization encoding.

Note that our proposed error-controlled quantization is totally different from the traditional quantization technique, \textit{vector quantization}, used in previous lossy compression, such as SSEM \cite{sasaki} and NUMARCK \cite{chen}, in two properties: uniformity and error-control. The vector quantization method is nonuniform, whereas our quantization is uniform. Specifically, in vector quantization, the more concentratedly the data locates, the shorter the quantization interval will be, while the length of our quantization intervals is fixed (i.e. twice the error bound). Therefore, in vector quantization, the compression error cannot be controlled for every data point, especially the points in the intervals with the length longer than twice the error bound. Thus, we call our quantization method as \textit {error-controlled quantization}.

\begin{figure}[t]
\centering
\includegraphics[scale=0.4]{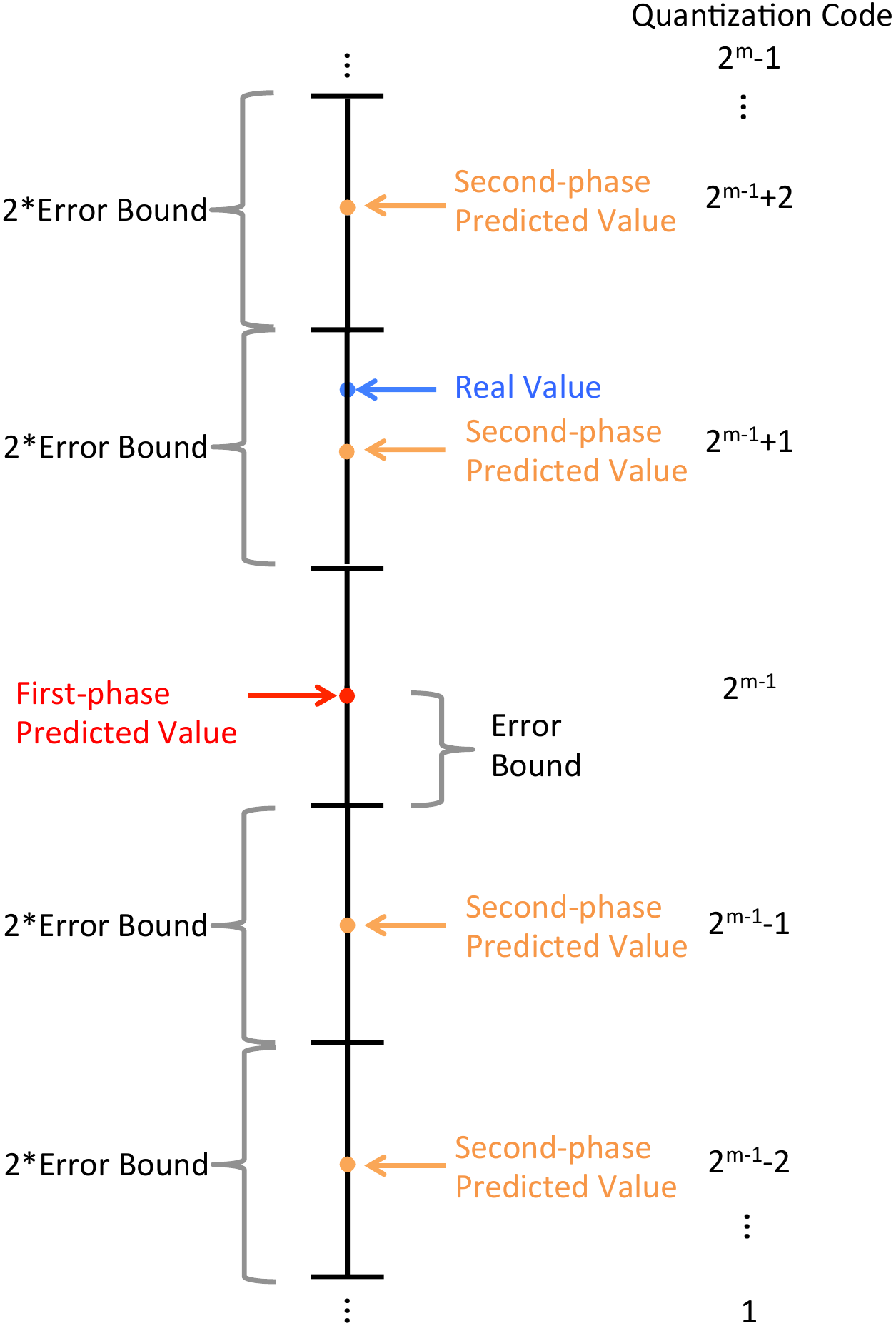}
\caption{Design of error-controlled quantization based on linear scaling of the error bound.}
\label{fig: f2}
\vspace{-4mm}
\end{figure}

The next question is, How many quantization intervals should we use in the error-controlled quantization? We leave this question to Subsection IV-B. First, we introduce a technique we will adopt after the quantization.

Figure \ref{fig: f3} shows an example of the distribution of quantization codes produced by our quantization encoder, which uses 255 quantization intervals to represent predictable data. From this figure, we see that the distribution of quantization codes is uneven and that the degree of nonuniformity of the distribution depends on the accuracy of the previous prediction. In information and coding theory, a strategy, called \textit {variable-length encoding}, is used to compress the nonuniform distribution source. In variable-length encoding, more common symbols will be generally represented using fewer bits than less common symbols. For uneven distribution, we can employ the variable-length encoding to reduce the data size significantly. Note that variable-length encoding is a process of lossless data compression.

Specifically, we use the most popular variable-length encoding strategy, \textit {Huffman coding}. Here we do not describe the Huffman coding algorithm in detail, but we note that Huffman coding algorithm implemented in all the lossless compressors on the market can deal only with the source byte by byte; hence, the total number of the symbols is as higher as to 256 ($2^8$). In our case, however, we do not limit $m$ to be no greater than 8. Hence, if $m$ is larger than 8, more than 256 quantization codes need to be compressed using the Huffman coding. Thus, in our compression, we implement a highly efficient Huffman coding algorithm that can handle a source with any number of quantization codes. % We will show how much data will be reduced after the use of variable-length encoding.

\subsection{Adaptive Scheme for Number of Quantization Intervals}
In Subsection IV-A, our proposed compression algorithm encodes the predictable data with its corresponding quantization code and then uses variable-length encoding to reduce the data size. A question remaining: How many quantization intervals should we use?

We use an $m-bit$ code to encode each data point, and the unpredictable data will be stored after a reduction of binary-representation analysis \cite{di}. However, even binary-representation analysis can reduce the data size to a certain extent. Storing the unpredictable data point has much more overhead than storing the quantization codes. Therefore, we should select a value for the number of quantization intervals that is as small as possible but can provide a sufficient prediction hitting rate. Note that the rate depends on the error bound as shown in Figure \ref{fig: f4}. If the error bound is too low (e.g., $eb_{rel} = 10^{-7}$), the compression is close to lossless, and achieving a high prediction hitting rate is difficult. Hence, we focus our research on a reasonable range of error bounds, $eb_{rel} \geq 10^{-6}$. 

\begin{figure}[t]
\centering
\includegraphics[scale=0.34]{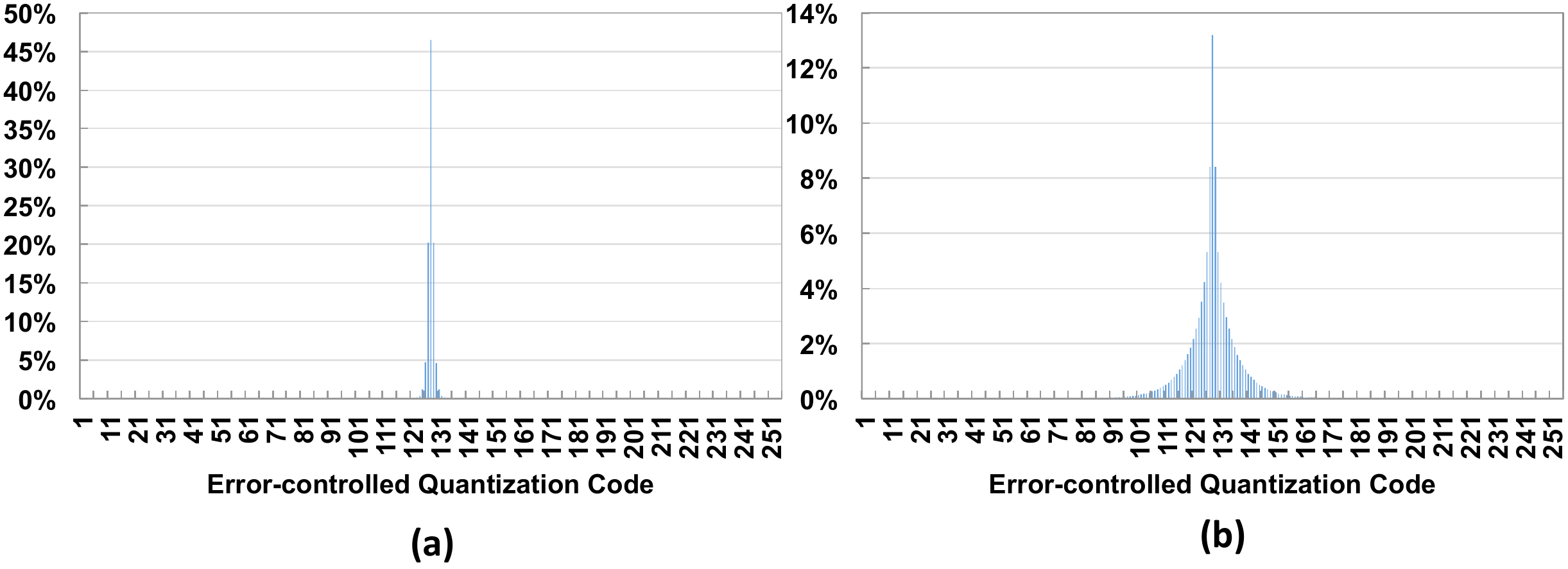}
\vspace{-3mm}
\caption{Distribution produced by error-controlled quantization encoder on ATM data sets of (a) value-range-based relative error bound = $10^{-3}$ and (b) value-range-based relative error bound = $10^{-4}$ with 255 quantization intervals ($m=8$).}
\label{fig: f3}
\vspace{-4mm}
\end{figure}

Now we introduce our adaptive scheme for the number of quantization intervals used in the compression algorithm. Figure \ref{fig: f4} shows the prediction hitting rate with different value-range-based relative error bounds using different numbers of quantization intervals on 2D ATM data sets and 3D hurricane data sets. It indicates that the prediction hitting rate will suddenly descend at a certain error bound from over 90\% to a relatively low value. For example, if using 511 quantization intervals, the prediction hitting rate will drop from 97.1\% to 41.4\% at $eb_{rel} = 10^{-6}$. Thus, we consider that 511 quantization intervals can cover only the value-range-based relative error bound higher than $10^{-6}$. However, different numbers of quantization intervals have different capabilities to cover different error bounds. Generally, more quantization intervals will cover lower error bounds. Baker et al. \cite{baker} point out that $eb_{rel} = 10^{-5}$ is enough for climate research simulation data sets, such as ATM data sets. Thus, based on Fig. \ref{fig: f4}, for ATM data sets, using 63 intervals and 511 intervals are good choices for $eb_{rel} = 10^{-4}$ and $eb_{rel} = 10^{-5}$ respectively.
%Thus, for ATM data sets, using 511 quantization intervals is a good choice.
But, for hurricane data sets, we suggest using 15 intervals for $eb_{rel} = 10^{-4}$ and 63 intervals for $eb_{rel} = 10^{-5}$. 
%using 255 quantization intervals is enough to cover $eb_{rel} \geq 10^{-5}$.

%Based on our experiments on a large variety of scientific data sets, using 511 quantization intervals can achieve a good prediction hitting rate within reasonable error bounds. Thus, in our compressor we set 511 quantization intervals ($m = 9$) as the default. 
In our compression algorithm, a user can determine the number of quantization intervals by setting a value for $m$ ($2^m-1$ quantization intervals). However, if it is unable to achieve a good prediction hitting rate (smaller than $\theta$) in some error bounds, our compression algorithm will suggest that the user increases the number of quantization intervals. On contrast, the user should reduce the number of quantization intervals until a further reduction results the prediction hitting rate smaller than $\theta$. In practice, sometimes a user's requirement for compression accuracy is stable; therefore, the user can tune a good value for the number of quantization intervals and get optimized compression factors in the following large-scale compression.

\begin{figure}[t]
\centering
\includegraphics[scale=0.35]{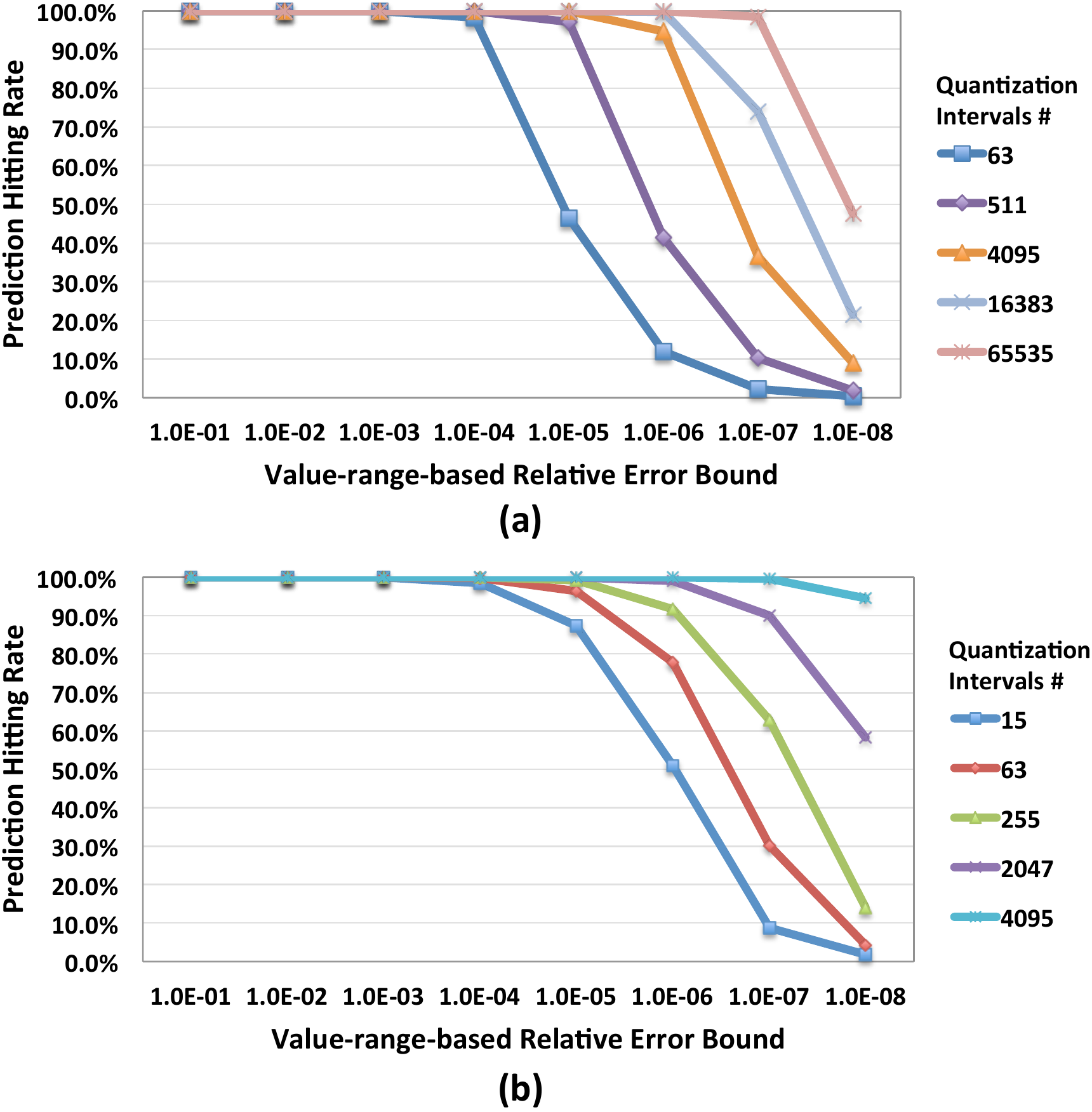}
\caption{Prediction hitting rate with decreasing error bounds using different quantization intervals on (a) ATM data sets and (b) hurricane data sets.}
\label{fig: f4}
\vspace{-4mm}
\end{figure}

\begin{figure}[t]
\centering
\includegraphics[scale=0.60]{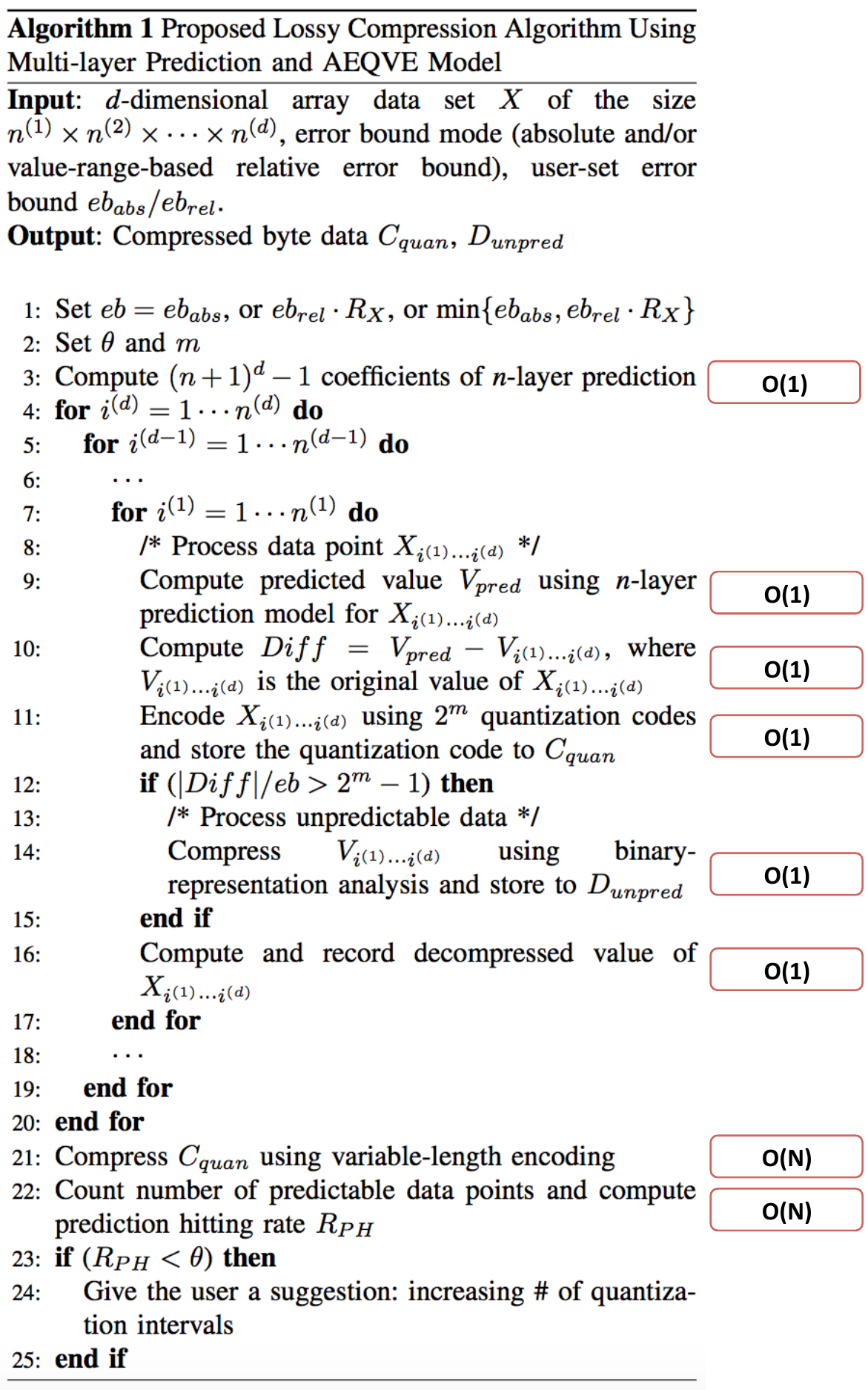}
\caption{Proposed lossy compression algorithm using Multi-layer Prediction and AEQVE Model}
\label{fig: f10}
\vspace{-4mm}
\end{figure}

\textbf {Algorithm 1} in Figure \ref{fig: f10} outlines our proposed lossy compression algorithm. Note that the input data is a \textit {d}-dimensional floating-point array of the size $N = n^{(1)} \times n^{(2)} \times \cdots \times n^{(d)}$, where $n^{(1)}$ is the size of the lowest dimension and $n^{(d)}$ is the size of the highest dimension. %In our algorithm, we compress the data from low dimension to high dimension.
Before processing the data (\textcolor{red} {line 1-3}), our algorithm needs to compute the $(n+1)^d-1$ coefficients (based on Equation \ref{eq: predmodel}) of the \textit {n}-layer prediction method only once (\textcolor{red}{line 3}). While processing the data (\textcolor{red} {line 4-20}), first, the algorithm computes the predicted value for the current processing data point using the \textit {n}-layer prediction method (\textcolor{red} {line 9}). Next, the algorithm computes the difference between the original and predicted data value and encodes the data point using $2^m$ quantization codes (\textcolor{red}{line 10-11}). Then, if the data point is unpredictable, the algorithm adopts the binary-representation analysis (\textcolor{red} {line 14}) proposed in \cite{di} to reduce its storage. Lastly, the algorithm computes and records the decompressed value for the future prediction (\textcolor{red} {line 16}). After processing each data point (\textcolor{red} {line 21-25}), the algorithm will compress the quantization codes using the variable-length encoding technique (\textcolor{red} {line 21}) and count the number of predictable data points (\textcolor{red} {line 22}). If the prediction hitting rate is lower than the threshold $\theta$, our algorithm will suggest that the user increases the quantization interval number (\textcolor{red} {line 23-25}). The computation complexity of each step is shown in Figure \ref{fig: f4}. Note that (1) lines 3 and 9 are $O(1)$, since they depend only on the number of layers $n$ used in the prediction rather than the data size $N$; (2) although line 14 is $O(1)$, binary-presentation analysis is more time-consuming than the other $O(1)$ operations, such as lines 9-11 and 16, and hence increasing the prediction hitting rate can result in faster compression significantly; and (3) since we adopt the Huffman coding algorithm for the variable-length encoding and the total number of the symbols (i.e., quantization intervals) is $2^m$ (such as 255), line 22 is its theoretical complexity $O$($N\log{}{2^m}$) = $O$($m$$N$) = $O$($N$). Therefore, the overall complexity is $O(N)$.
\section{Empirical Performance Evaluation}
\label{sec: evaluation}

In this section, we evaluate our compression algorithm, namely SZ-1.4, on various single-precision floating-point data sets: 2D \textit{ATM} data sets from climate simulations \cite{cesm}, 2D \textit{APS} data sets from X-ray scientific research \cite{aps}, and 3D \textit{hurricane} data sets from a hurricane simulation \cite{hurricane}, as shown in Table \ref{tab: data}. Also, we compare our compression algorithm SZ-1.4 with state-of-the-art losseless (i.e., GZIP \cite{gzip} and FPZIP \cite{fpzip}) and lossy compressors (i.e., ZFP \cite{lindstrom}, SZ-1.1 \cite{di}, and ISABELA \cite{laksh}), based on the metrics mentioned in Section III. We conducted our experiments on a single core of an iMac with 2.3 GHz Intel Core i7 processors and 32 GB of 1600 MHz DDR3 RAM.

\begin{table}[]
\centering
\caption{Description of data sets used in empirical performance evaluation}
\label{tab: data}
\begin{adjustbox}{max width=0.48\textwidth}
\begin{tabular}{|c|c|c|c|c|}
\hline
          & Data Source               & Dimension Size & Data Size & File Number \\  \hline
\textbf{ATM}       & Climate simulation   & $1800 \times 3600$        & $2.6$ TB & 11400\\ \hline
\textbf{APS}       & X-ray instrument     & $2560 \times 2560$         & $40$ GB & 1518 \\ \hline
\textbf{Hurricane} & Hurricane simulation & $100 \times 500 \times 500$        & $1.2$ GB & 624 \\ \hline
\end{tabular}
\end{adjustbox}
\vspace{-6mm}
\end{table}

\begin{table*}[t]\scriptsize
\centering
\caption{Comparison of Pearson correlation coefficient using various lossy compressors with different maximum compression errors}
\vspace{-2mm}
\label{tab: correlation}
\begin{tabular}{|c|c|c|c|c|c|c|c|c|}
\hline
\multirow{2}{*}{\textbf{\begin{tabular}[c]{@{}c@{}}Maximum\\ $e_{rel}$\end{tabular}}} & \multicolumn{3}{c|}{\textbf{ATM}}                                & \multirow{7}{*}{} & \multirow{2}{*}{\textbf{\begin{tabular}[c]{@{}c@{}}Maximum \\ $e_{rel}$\end{tabular}}} & \multicolumn{3}{c|}{\textbf{Hurricane}}                          \\ \cline{2-4} \cline{7-9} 
                                                                                                         & \textit{SZ-1.4} & \textit{ZFP}      & \textit{SZ-1.1}       &                   &                                                                                                          & \textit{SZ-1.4} & \textit{ZFP}      & \textit{SZ-1.1}       \\ \cline{1-4} \cline{6-9} 
$3.3 \times 10^{-3}$                                                                                     & $0.99998$                 & $0.9996$          & $0.99998$          &                   & $2.4 \times 10^{-3}$                                                                                     & $0.998$                 & $0.99995$          & $0.998$          \\ \cline{1-4} \cline{6-9} 
$4.3 \times 10^{-4}$                                                                                     & $\geq 1 - 10^{-6}$        & $\geq 1 - 10^{-7}$ & $\geq 1 - 10^{-6}$          &                   & $1.8 \times 10^{-4}$                                                                                     & $\geq 1 - 10^{-5}$        & $\geq 1 - 10^{-6}$ & $\geq 1 - 10^{-5}$          \\ \cline{1-4} \cline{6-9} 
$2.6 \times 10^{-5}$                                                                                     & $\geq 1 - 10^{-8}$        & $\geq 1 - 10^{-9}$ & $\geq 1 - 10^{-9}$ &                   & $2.5 \times 10^{-5}$                                                                                     & $\geq 1 - 10^{-6}$        & $\geq 1 - 10^{-8}$ & $\geq 1 - 10^{-5}$ \\ \cline{1-4} \cline{6-9} 
$3.4 \times 10^{-6}$                                                                                     & $\geq 1 - 10^{-10}$        & $\geq 1 - 10^{-11}$ & $\geq 1 - 10^{-11}$ &                   & $2.6 \times 10^{-6}$                                                                                     & $\geq 1 - 10^{-8}$        & $\geq 1 - 10^{-9}$ & $\geq 1 - 10^{-7}$ \\ \cline{1-4} \cline{6-9} 
$4.1 \times 10^{-7}$                                                                                     & $\geq 1 - 10^{-12}$        & $\geq 1 - 10^{-13}$ & $\geq 1 - 10^{-13}$ &                   & $2.9 \times 10^{-7}$                                                                                     & $\geq 1 - 10^{-10}$        & $\geq 1 - 10^{-11}$ & $\geq 1 - 10^{-11}$ \\ \hline
\end{tabular}
% \begin{tabular}{|c|c|c|c|c|c|c|c|c|c|c|c|c|}
% \hline
% \multirow{2}{*}{\textbf{$eb_{rel}$}} & \textbf{ATM}             &              &             &                  & \multicolumn{4}{c|}{\textbf{APS}}                                        & \multicolumn{4}{c|}{\textbf{Hurricane}}          \\ \cline{2-13} 
%                                      & \textit{Our compression} & \textit{ZFP} & \textit{SZ} & \textit{ISABELA} & \textit{Our compression} & \textit{ZFP} & \textit{SZ} & \textit{ISABELA} & Our compression & ZFP      & SZ       & ISABELA  \\ \hline
% $10^{-3}$                            & 0.999998                 & 0.999998     & 0.999913    & 1.000000         & 0.999887                 & 1.000000     & 0.999983    & 0.999995         & 0.999582        & 0.999998 & 0.999783 & 1.000000 \\ \hline
% $10^{-4}$                            & 1.000000                 & 1.000000     & 1.000000    & 1.000000         & 0.999999                 & 1.000000     & 1.000000    & 0.999998         & 0.999996        & 0.999999 & 0.999783 & 1.000000 \\ \hline
% $10^{-5}$                            & 1.000000                 & 1.000000     & 1.000000    & 1.000000         & 1.000000                 & 1.000000     & 1.000000    & 1.000000         & 1.000000        & 1.000000 & 1.000000 & /        \\ \hline
% $10^{-6}$                            & 1.000000                 & 1.000000     & 1.000000    & /                & 1.000000                 & 1.000000     & 1.000000    & /                & 1.000000        & 1.000000 & 1.000000 & /        \\ \hline
% \end{tabular}
\vspace{-1.0\baselineskip}
\end{table*}

\subsection{Compression Factor}
First, we evaluated our compression algorithm (i.e., SZ-1.4) based on the \textit {compression factor}. Figure \ref{fig: f5} compares the compression factors of SZ-1.4 and five other compression methods: GZIP, FPZIP, ZFP, SZ-1.1, and ISABELA, with reasonable value-range-based relative error bounds, namely, $10^{-3}$, $10^{-4}$, $10^{-5}$, and $10^{-6}$, respectively. Specifically, we ran different compressors using the absolute error bounds computed based on the above listed ratios and the global data value range and then checked the compression results. Figure \ref{fig: f5} indicates that SZ-1.4 has the best compression factor within these reasonable error bounds. For example, with $eb_{rel} = 10^{-4}$, for ATM data sets, the average compression factor of SZ-1.4 is 6.3, which is 110\% higher than ZFP's 3.0, 70\% higher than SZ-1.1's 3.8, 350\% higher than ISABELA's 1.4, 232\% higher than FPZIP's 1.9, and 430\% higher than GZIP's 1.3. For APS data sets, the average compression factor of SZ-1.4 is 5.2, which is 79\% higher than ZFP's 2.9, 74\% higher than SZ-1.1's 3.0, 340\% higher than ISABELA 1.2, 300\% higher than FPZIP's 1.3, and 372\% higher than GZIP's 1.1. For the hurricane data sets, the average compression factor of SZ-1.4 is 21.3, which is 166\% higher than ZFP's 8.0, 139\% higher than SZ-1.1's 8.9, 1675\% higher than ISABELA's 1.2, 788\% higher than FPZIP's 2.4, and 1538\% higher than GZIP's 1.3. Note that ISABELA cannot deal with some low error bounds; thus, we plot its compression factors only until it fails.

\begin{figure}[t]
\centering
\includegraphics[scale=0.55]{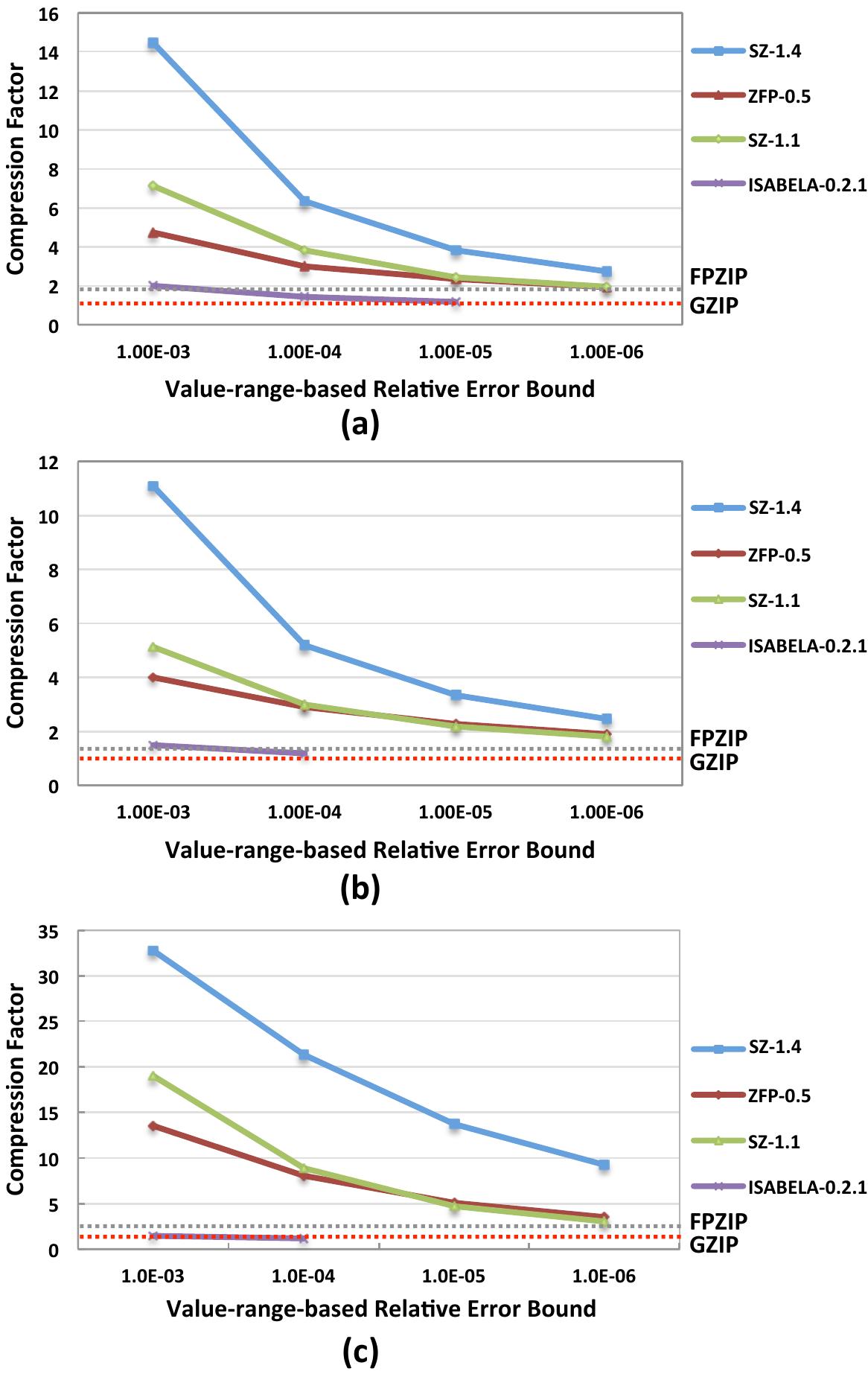}
\vspace{-2mm}
\caption{Comparison of compression factors using different lossy compression methods on (a) ATM, (b) APS, and (c) hurricane data sets with different error bounds.}
\label{fig: f5}
\vspace{-4mm}
\end{figure}

\begin{table}[]
\centering
\caption{Maximum compression errors (normalized to value range) using SZ-1.4 and ZFP with different user-set value-range-based error bounds}
\label{tab: maxerror}
\begin{adjustbox}{max width=0.48\textwidth}
\begin{tabular}{|c|c|c|c|c|}
\hline
\multirow{2}{*}{\textbf{User-set $eb_{rel}$}} & \multicolumn{2}{c|}{\textbf{ATM}}              & \multicolumn{2}{c|}{\textbf{Hurricane}}          \\ \cline{2-5} 
                                              & \textit{SZ-1.4} & \textit{ZFP}        & \textit{SZ-1.4} & \textit{ZFP}          \\ \hline
$10^{-2}$                                     & $1. 0 \times 10^{-2}$    & $3.3\times 10^{-3}$ & $1. 0 \times 10^{-2}$    & $2.4\times 10^{-3}$   \\ \hline
$10^{-3}$                                     & $1. 0 \times 10^{-3}$    & $4.3\times 10^{-4}$ & $1. 0 \times 10^{-3}$    & $1.8 \times 10^{-4}$ \\ \hline
$10^{-4}$                                     & $1. 0 \times 10^{-4}$    & $2.6\times 10^{-5}$ & $1. 0 \times 10^{-4}$    & $2.5\times 10^{-5}$   \\ \hline
$10^{-5}$                                     & $1. 0 \times 10^{-5}$    & $3.4\times 10^{-6}$ & $1. 0 \times 10^{-5}$    & $2.6\times 10^{-6}$   \\ \hline
$10^{-6}$                                     & $1. 0 \times 10^{-6}$    & $4.1\times 10^{-7}$ & $1. 0 \times 10^{-6}$    & $2.9\times 10^{-7}$   \\ \hline
\end{tabular}
\end{adjustbox}
\vspace{-4mm}
\end{table}

\begin{table*}[t]\scriptsize
\centering
\caption{Compression and decompression speeds (MB/s) using SZ-1.4 and ZFP with different value-range-based relative error bounds}
\vspace{-2mm}
\label{tab: speedfactor}
\begin{tabular}{|c|c|c|c|c|c|c|c|c|c|c|c|c|}
\hline
\multirow{3}{*}{\textbf{User-set $eb_{rel}$}} & \multicolumn{4}{c|}{\textbf{ATM}}                                                 & \multicolumn{4}{c|}{\textbf{APS}}                                                 & \multicolumn{4}{c|}{\textbf{Hurricane}}                                           \\ \cline{2-13} 
                                              & \multicolumn{2}{c|}{\textit{SZ-1.4}} & \multicolumn{2}{c|}{\textit{ZFP}} & \multicolumn{2}{c|}{\textit{SZ-1.4}} & \multicolumn{2}{c|}{\textit{ZFP}} & \multicolumn{2}{c|}{\textit{SZ-1.4}} & \multicolumn{2}{c|}{\textit{ZFP}} \\ \cline{2-13} 
                                              & Comp                 & Decomp                 & Comp            & Decomp          & Comp                 & Decomp                 & Comp            & Decomp          & Comp                  & Decomp                & Comp            & Decomp          \\ \hline
$10^{-3}$                                     & 82.3                  & 174.0                  & 118.7           & 181.8           & 77.7                  & 130.5                   & 101.1           & 156.5           & 84.9                    & 176.4                   & 251.6            & 549.6            \\ \hline
$10^{-4}$                                     & 61.5                  & 100.6                   & 100.5            & 139.4           & 64.3                  & 98.0                   & 104.5            & 133.6           & 82.8                    & 164.5                    & 211.3            & 436.0             \\ \hline
$10^{-5}$                                     & 55.4                  & 83.8                   & 87.9            & 121.3            & 52.9                  & 78.8                   & 101.7            & 115.3           & 76.2                    & 149.0                   & 174.3             & 322.8             \\ \hline
$10^{-6}$                                     & 46.1                  & 55.6                   & 83.6            & 105.7            & 44.3                  & 50.8                   & 95.4            & 109.7            & 69.5                    & 118.1                   & 150.9             & 265.4            \\ \hline
\end{tabular}
\vspace{-4mm}
\end{table*}

We note that ZFP might not respect the error bound because of the fixed-point alignment when the value range is huge. For example, the variable \texttt{CDNUMC} in the ATM data sets, its value range is from $10^{-3}$ to $10^{11}$ and the compression error of the data point with the value $6.936168$ is $0.123668$ if using ZFP with $eb_{abs} = 10^{-7}$. When the value range is not such huge, the maximum compression error of ZFP is much lower than the input error bound, whereas the maximum compression errors of the other lossy compression methods, including SZ-1.4, are exactly the same as the input error bound. This means that ZFP is overconservative with regard to the user's accuracy requirement. Table \ref{tab: maxerror} shows the maximum compression errors of SZ-1.4 and ZFP with different error bounds. For a fair comparison, we also evaluated SZ-1.4 by setting its input error bound as the maximum compression error of ZFP, which will make the maximum compression errors of SZ-1.4 and ZFP the same. The comparison of compression factors is shown in Figure \ref{fig: f6}. For example, with the same maximum compression error of $4.3 \times 10^{-4}$, our average compression factor is 162\% higher than ZFP's on the ATM data sets. With the same maximum compression error of $1.8 \times 10^{-4}$, our average compression factor is 71\% higher than ZFP's on the hurricane data sets.

\begin{figure}[t]
\centering
\includegraphics[scale=0.35]{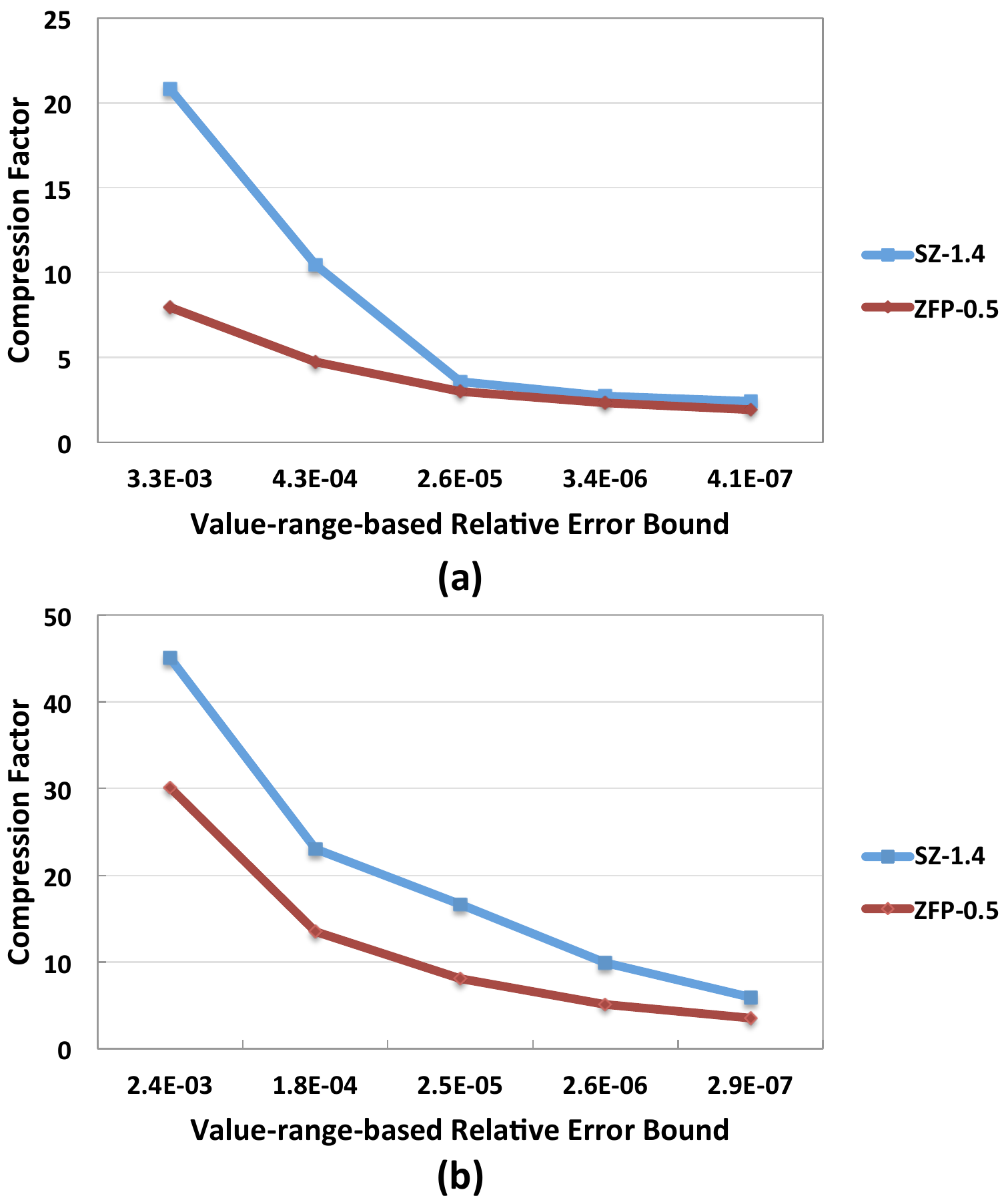}
\vspace{-2mm}
\caption{Comparison of compression factors with same maximum compression error using SZ-1.4 and ZFP on (a) ATM and (b) hurricane data sets.}
\label{fig: f6}
\vspace{-6mm}
\end{figure}

\subsection{Rate-Distortion}
We note that ZFP is designed for a fixed bit-rate, whereas SZ (including SZ-1.1 and SZ-1.4) and ISABELA are designed for a fixed maximum compression error. Thus, for a fair comparison, we plot the \textit{rate-distortion} curve for all the lossy compressors and compare the distortion quality with the same rate. Here rate means bit-rate in bits/value, and we will use the peak signal-to-noise ratio (PSNR) to measure the distortion quality. PSNR is calculated by the equation (\ref{eq: psnr}) in decibel. Generally speaking, in the rate-distortion curve, the higher the bit-rate (i.e., more bits per value) in compressed storage, the higher the quality (i.e., higher PSNR) of the reconstructed data after decompression. 

Figure \ref{fig: f7} shows the rate-distortion curves of the different lossy compressors on the three scientific data sets. The figure indicates that our lossy compression algorithm (i.e., SZ-1.4) has the best rate-distortion curve on the 2D data sets, ATM and APS. Specifically, when the bit-rate equals 8 bits/value (i.e., $CF=4$), for the ATM data sets, the PSNR of SZ-1.4 is about 103 dB, which is 14 dB higher than the second-best ZFP's 89 dB. This 14 dB improvement in PSNR represents an increase in accuracy (or reduction in RMSE) of more than 5 times. Also, the accuracy of our compressor is more than 7 times that of SZ-1.1 and $10^{3}$ times than of ISABELA. For APS data sets, the PSNR of SZ-1.4 is about 96 dB, which is 9 dB higher than ZFP's 87 dB. This 9 dB improvement in PSNR represents an increase in accuracy of 2.8 times. Also, the accuracy of our compressor is 8 times that of SZ-1.1 and 790 times that of ISABELA.

For the 3D hurricane data sets, the rate-distortion curves illustrate that at low bit-rate (i.e., 2 bits/value) the PSNR of SZ-1.4 is close to that of ZFP. In the other cases of bit-rate higher than 2 bits/value, our PSNR is better than ZFP's. Specifically, when the bit-rate is 8 bits/value, our PSNR is about 182 dB, which is 11 dB higher (i.e., 3.5 times in accuracy) than ZFP's 171 dB, and 47 dB higher (i.e., 224 times in accuracy) than SZ-1.1's 135 dB.

Note that we test and show the cases only with the bit-rate lower than 16 bits/value for the three single-precision data sets, which means the compression factors are higher than 2. As we mentioned in Section I, some lossless compressors can provide a compression factor up to 2 \cite{ratana}. It is reasonable to assume that users are interested in lossy compression only if it provides a compression factor of 2 or higher.

\begin{figure}[t]
\centering
\includegraphics[scale=0.55]{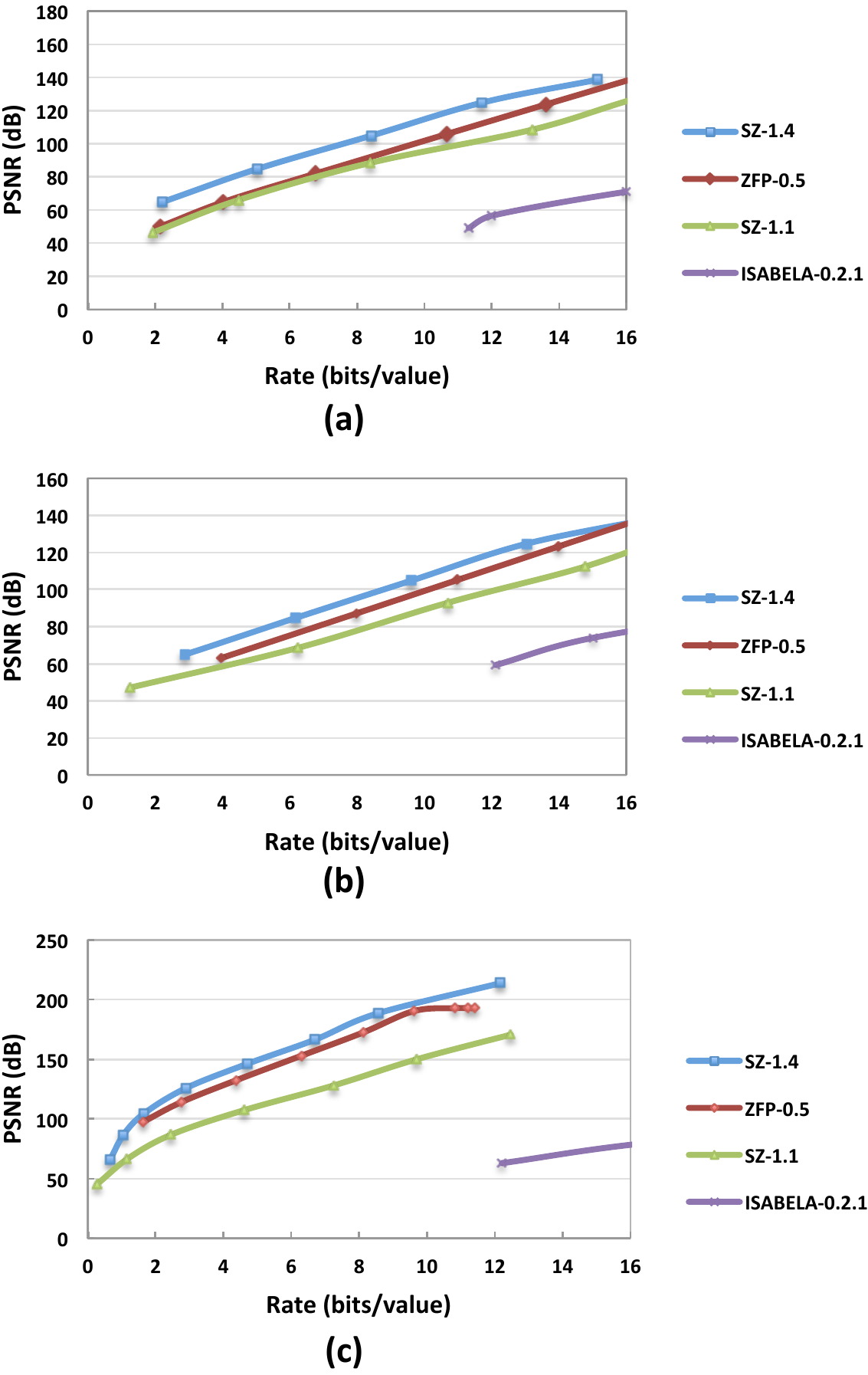}
\vspace{-2mm}
\caption{Rate-distortion using different lossy compression methods on (a) ATM, (b) APS, and (c) hurricane data sets.}
\label{fig: f7}
\vspace{-4mm}
\end{figure}

\subsection{Pearson Correlation}
Next we evaluated our compression algorithm (i.e., SZ-1.4) based on the Pearson correlation coefficient between the original and the decompressed data. Table \ref{tab: correlation} shows the Pearson correlation coefficients using different lossy compression methods with different maximum compression errors. Because of space limitations, we compare SZ-1.4 only with ZFP and SZ-1.1, since from the previous evaluations they outperform ISABELA significantly. We note that we use the maximum compression error of ZFP as the input error bound of SZ-1.4 and SZ-1.1 to make sure that all three lossy compressors have the same maximum compression error. From Table \ref{tab: correlation} we know that all three compressors have ``five nines" or better coefficients (marked with bold) (1) from $4.3 \times 10^{-4}$ to lower value-range-based relative error bounds on the ATM data sets and (2) from $1.8 \times 10^{-4}$ to lower value-range-based relative error bounds on the hurricane data sets. These results mean SZ-1.4 has accuracy in the Pearson correlation of decompressed data similar to that of ZFP and SZ-1.1. 

\subsection{Speed}
Now, let us evaluate the compression and decompression speed of our compressor (i.e., SZ-1.4). %Table \ref{tab: speed} shows the overall speed of different lossy compressors with different error bounds in megabytes per second, including compression and decompression times together. For the 2D ATM and APS data sets, on average, our compressor is 3\% faster than ZFP, 2.2x faster than SZ, and 32x faster than ISABELA. For the 3D hurricane data sets, on average, our compressor is 34.4\% slower than ZFP, 2.4x faster than SZ, and 62x faster than ISABELA. 
We evaluate the compression and decompression speed of different lossy compressors with different error bound in megabytes per second. First, we compare the overall speed of SZ-1.4 with SZ-1.1 and ISABELA's. For the 2D ATM and APS data sets, on average, our compressor is 2.2x faster than SZ-1.1 and 32x faster than ISABELA. For the 3D hurricane data sets, on average, SZ-1.4 is 2.4x faster than SZ-1.1 and 62x faster than ISABELA. Due to space limitations, we do not show the specific values of SZ-1.1 and ISABELA. We then compare the speed of SZ-1.4 and ZFP. Table \ref{tab: speedfactor} shows the compression and decompression speed of SZ-1.4 and ZFP. It illustrates that on average SZ-1.4's compression is 50\% slower than ZFP's and decompression is 48\% slower than ZFP's. Our compression has not been optimized in performance because the primary objective was to reach high compression factors, therefore, we plan to optimize our compression for different architectures and data sets in the future.

\subsection{Autocorrelation of Compression Error}
%Figure \ref{fig: f8} shows the distribution of the compression errors of our lossy compressor and ZFP with the same maximum compression error on ATM and Hurricane data sets. It indicates that the distribution of our compression errors is nearly uniform, while the distribution of ZFP is similar to normal distribution. 
Finally, we analyze the autocorrelation of the compression errors, since some applications require the compression errors to be uncorrelated. %or white noise \cite{wu}. 
We evaluate the autocorrelation of the compression errors on the two typical variables in the ATM data sets, i.e., \texttt{FREQSH} and \texttt{SNOWHLND}. The compression factors of \texttt{FREQSH} and \texttt{SNOWHLND} are $6.5$ and $48$ using SZ-1.4 with $eb_{rel} = 10^{-4}$. Thus, to some extent, \texttt{FREQSH} can represent relatively low-compression-factor data sets, while \texttt{SNOWHLND} can represent relatively high-compression-factor data sets. Figure \ref{fig: f9} shows the first 100 autocorrelation coefficients of our and ZFP's compression errors on these two variables. It illustrates that on the \texttt{FREQSH} the maximum autocorrelation coefficient of SZ-1.4 is $4 \times 10^{-3}$, which is much lower than ZFP's $0.25$. However, on the \texttt{SNOWHLND} the maximum autocorrelation coefficient of SZ-1.4 is about $0.5$, which is higher than ZFP's $0.23$. We also evaluate the autocorrelation of SZ-1.4 and ZFP on the APS and hurricane data sets and observe that, generally, SZ-1.4's autocorrelation is lower than ZFP's on the relatively low-compression-factor data sets, whereas ZFP's autocorrelation is lower than SZ-1.4's on the relatively high-compression-factor data sets. We therefore plan to improve the autocorrelation of compression errors on the relatively high-compression-factor data sets in the future. The effect of compression error autocorrelation being application specific, lossy compressor users might need to understand this effect before using one of the other compressor.

%We assume that the compression errors in the data set follow a stationary process. 
%Figure \ref{fig: f9} shows the first 100 autocorrelation coefficients of our and ZFP's compression errors on these two variables. It illustrates that the maximum autocorrelation coefficient of our compression errors is only about $6 \times 10^{-3}$ on the ATM data sets and $3 \times 10^{-3}$ on the hurricane data sets, which is lower than ZFP's $0.25$ and $0.15$. Thus, our compression errors can be considered as uncorrelated on the ATM and hurricane data sets, while ZFP's compression errors exhibits a higher autocorrelation on these data sets. The effect of compression error autocorrelation being application specific, lossy compressor users might need to understand this effect before using one of the other compressor.
%our compression errors can be considered as uncorrelated or white noise on the ATM and hurricane data sets, but %the correlation of ZFP's compression errors is much stronger 
%ZFP's compression errors seem to be not exactly white noise on the two data sets, which may result in a side effect on the applications requiring white noise errors.
%are not white noise on these two data sets and might have a side effect on the applications that must need white noise compression errors.

% \begin{figure}[t]
% \centering
% \includegraphics[scale=0.30]{Fig/Fig8.pdf}
% \caption{Distribution of compression errors using our lossy compressor and ZFP on (a) ATM and (b) Hurricane data sets.}
% \label{fig: f8}
% \vspace{-1.0\baselineskip}
% \end{figure}

\begin{figure}[t]
\centering
\includegraphics[scale=0.36]{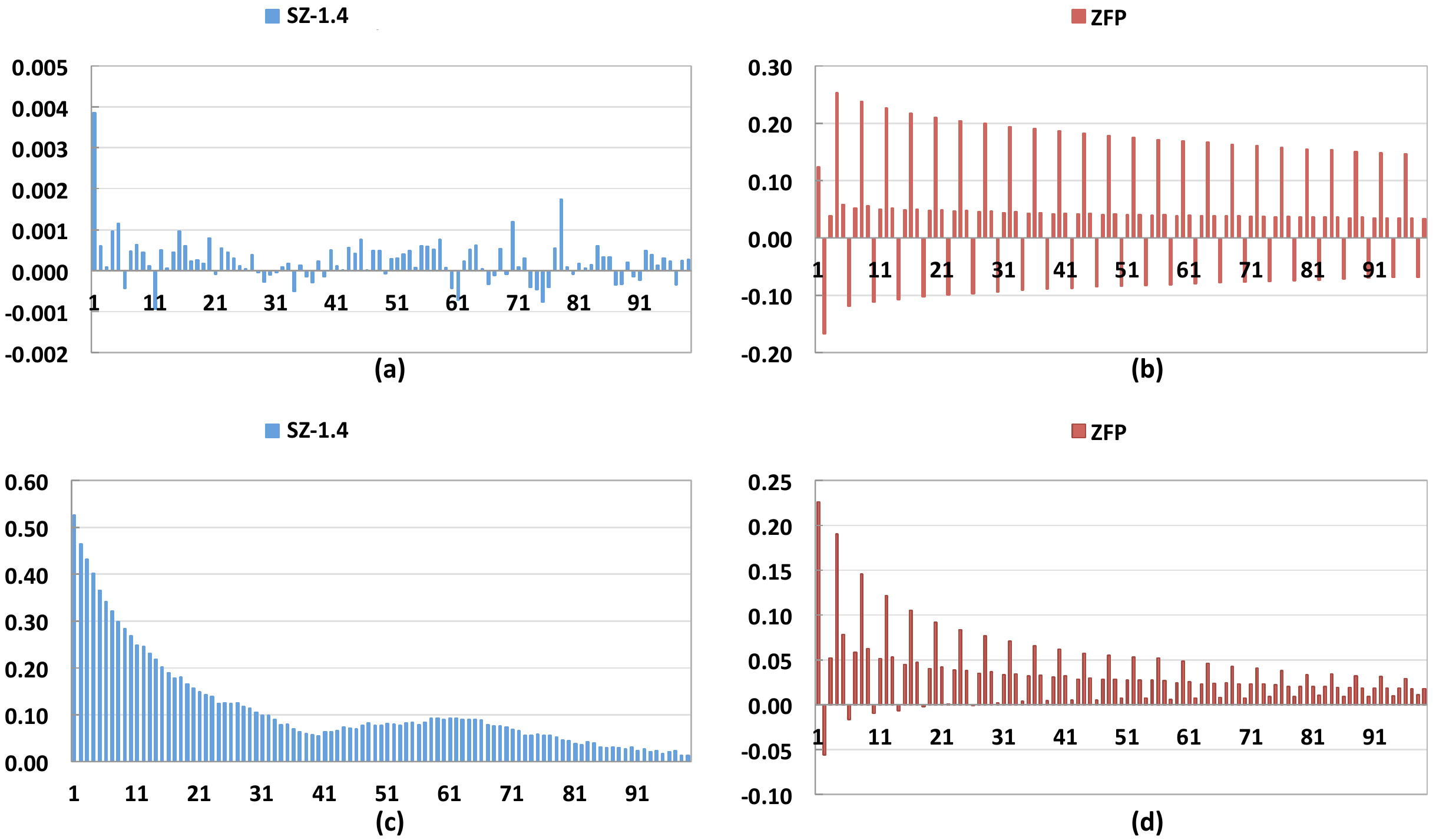}
\caption{Autocorrelation analysis (first 100 coefficients) of compression errors with increasing delays using our lossy compressor and ZFP on variable \texttt{FREQSH} (i.e., (a) and (b)) and variable \texttt{SNOWHLND} (i.e., (c) and (d)) in ATM data sets.}
\label{fig: f9}
\vspace{-2mm}
\end{figure}

\section{Discussion}
\label{sec: discuss}
In this section, we first discuss the parallel use of our compressor (i.e., SZ-1.4) for large-scale data sets. We then perform an empirical performance evaluation on the full 2.5 TB ATM data sets using 1024 cores (i.e., 64 nodes, each node with two Intel Xeon E5-2670 processors and 64 GB DDR3 memory, and each processor has 8 cores) from the Blues cluster at Argonne.

Parallel compression can be classified into two categories: in-situ compression and off-line compression. Our compressor can be easily used as an in-situ compressor embedded in a parallel application. Each process can compress/decompress a fraction of the data that is being held in its memory. For off-line compression, an MPI program or a script can be used to load the data into multiple processes and run the compression separately on them. ATM data sets (as shown in Table \ref{tab: data}), for example, have a total of 11400 files and APS data sets have 1518 files. The users can load these files by multiple processes and run our compressor in parallel, \textit {without inter-process communications}.

\begin{figure}[t]
\centering
\includegraphics[scale=0.46]{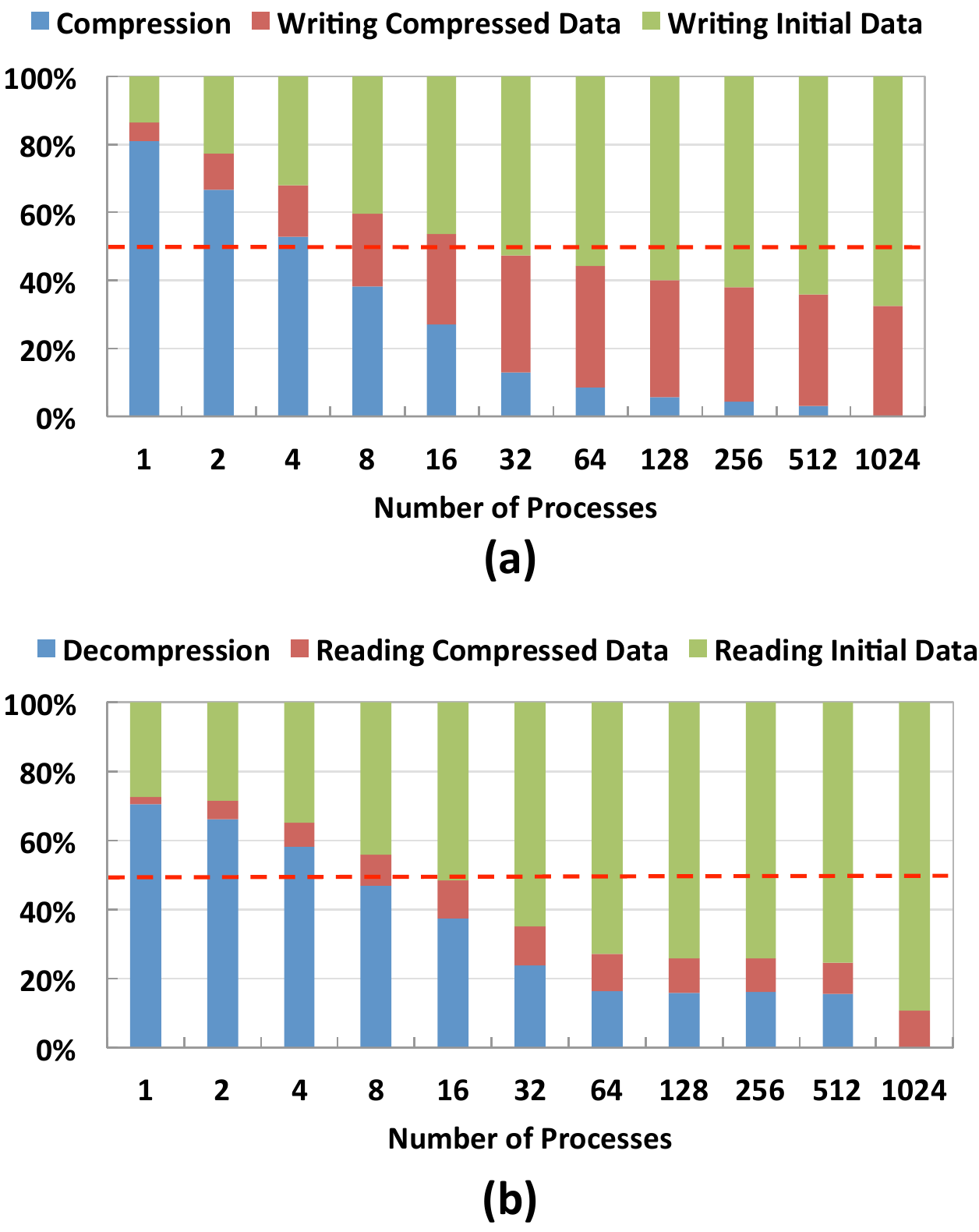}
\vspace{-4mm}
\caption{Comparison of time to compress/decompress and write/read compressed data against time to write/read initial data on Blues.}
\label{fig: f12}
\vspace{-4mm}
\end{figure}

\begin{table}[]
\centering
\caption{Strong scalability of parallel compression using SZ-1.4 with different number of processes on Blues}
\label{tab: parallelcomp}
\begin{adjustbox}{max width=0.37\textwidth}

\begin{tabular}{|c|c|c|c|c|}
\hline
\textbf{\begin{tabular}[c]{@{}c@{}}Number of\\ Processes\end{tabular}} & \textbf{\begin{tabular}[c]{@{}c@{}}Number of\\ Nodes\end{tabular}} & \textbf{\begin{tabular}[c]{@{}c@{}}Comp Speed\\ (GB/s)\end{tabular}} & \textbf{Speedup} & \textbf{\begin{tabular}[c]{@{}c@{}}Parallel\\ Efficiency\end{tabular}} \\ \hline
1                                                                      & 1                                                                  & 0.09                                                                 & 1.00             & 100.0\%                                                                \\ \hline
2                                                                      & 2                                                                  & 0.18                                                                 & 2.00             & 99.8\%                                                                 \\ \hline
4                                                                      & 4                                                                  & 0.35                                                                 & 3.99             & 99.9\%                                                                 \\ \hline
8                                                                      & 8                                                                  & 0.70                                                                 & 7.99             & 99.8\%                                                                 \\ \hline
16                                                                     & 16                                                                 & 1.40                                                                 & 15.98            & 99.9\%                                                                 \\ \hline
32                                                                     & 32                                                                 & 2.79                                                                 & 31.91            & 99.7\%                                                                 \\ \hline
64                                                                     & 64                                                                 & 5.60                                                                 & 63.97            & 99.9\%                                                                 \\ \hline
128                                                                    & 64                                                                 & 11.2                                                                 & 127.6            & 99.7\%                                                                 \\ \hline
256                                                                    & 64                                                                 & 21.5                                                                 & 245.8            & 96.0\%                                                                 \\ \hline
512                                                                    & 64                                                                 & 40.5                                                                 & 463.0            & 90.4\%                                                                 \\ \hline
1024                                                                   & 64                                                                 & 81.3                                                                 & 930.7            & 90.9\%                                                                 \\ \hline
\end{tabular}
\end{adjustbox}
\vspace{-2mm}
\end{table}

\begin{table}[]
\centering
\caption{Strong scalability of parallel decompression using SZ-1.4 with different number of processes on Blues}
\label{tab: paralleldecomp}
\begin{adjustbox}{max width=0.37\textwidth}

\begin{tabular}{|c|c|c|c|c|}
\hline
\textbf{\begin{tabular}[c]{@{}c@{}}Number of\\ Processes\end{tabular}} & \textbf{\begin{tabular}[c]{@{}c@{}}Number of\\ Nodes\end{tabular}} & \textbf{\begin{tabular}[c]{@{}c@{}}Decomp Speed\\ (GB/s)\end{tabular}} & \textbf{Speedup} & \textbf{\begin{tabular}[c]{@{}c@{}}Parallel\\ Efficiency\end{tabular}} \\ \hline
1                                                                      & 1                                                                  & 0.20                                                                   & 1.00             & 100.0\%                                                                \\ \hline
2                                                                      & 2                                                                  & 0.40                                                                   & 1.99             & 99.6\%                                                                 \\ \hline
4                                                                      & 4                                                                  & 0.80                                                                   & 4.00             & 99.9\%                                                                 \\ \hline
8                                                                      & 8                                                                  & 1.60                                                                   & 7.94             & 99.2\%                                                                 \\ \hline
16                                                                     & 16                                                                 & 3.20                                                                   & 16.00            & 99.9\%                                                                 \\ \hline
32                                                                     & 32                                                                 & 6.40                                                                   & 31.91            & 99.7\%                                                                 \\ \hline
64                                                                     & 64                                                                 & 12.8                                                                   & 64.00            & 99.9\%                                                                 \\ \hline
128                                                                    & 64                                                                 & 25.6                                                                   & 127.7            & 99.7\%                                                                 \\ \hline
256                                                                    & 64                                                                 & 49.0                                                                   & 244.5            & 95.5\%                                                                 \\ \hline
512                                                                    & 64                                                                 & 92.5                                                                   & 461.4            & 90.1\%                                                                 \\ \hline
1024                                                                   & 64                                                                 & 187.0                                                                  & 932.7            & 91.1\%                                                                 \\ \hline
\end{tabular}
\end{adjustbox}
\vspace{-4mm}
\end{table}

We present the strong scalability of the parallel compression and decompression without the I/O (i.e., writing/reading data) time in Table \ref{tab: parallelcomp} and \ref{tab: paralleldecomp} with different scales ranging from 1 to 1024 processes on the Blues cluster. In the experiments, we set $eb_{rel} = 10^{-4}$ for all the compression. The number of processes is increased in two stages. At the first stage, we launch one process per node and increase the number of nodes until the maximum number we can request (i.e., 64). At the second stage, we run the parallel compression on 64 nodes while changing the number of processes per node. We measure the time of compression/decompression without the I/O time and use the maximum time among all the processes. We test each experiment five times and use the average compression/decompression time to calculate their speeds, speedup, and parallel efficiency as shown in the tables. The two tables illustrates that the parallel efficiency of our compressor can stay nearly 100\% from 1 to 128 processes, which demonstrates that our compression/decompression have linear speedup with the number of processors. However, the parallel efficiency is decreased to about 90\% when the total number of processes is greater than 128 (i.e, more than two processes per node). This performance degradation is due to node internal limitations.
%This performance degradation is probably due to cache miss. Specifically, based on the prediction model we discussed in Section \ref{sec: prediction}, each data point needs to be reused multiple times for calculating the predicted values of its neighbored data points. Note that when the scale is greater than 128 processes in our experiment, some processors may hold at least two process, because there are at most 64 nodes (each facilitated with 2 processors) in our experiment. In this situation, the data of one process in the cache may be flushed out by the data of another process easily, leading to a cache miss when the data flushed out is needed again. 
Note that the compression/decompression speeds of a singe process in Table \ref{tab: parallelcomp} and \ref{tab: paralleldecomp} are different from ones in Table \ref{tab: speedfactor}, since we run the sequential and parallel compression on two different platforms.

Figure \ref{fig: f12} compares the time to compress/decompress and write/read the compressed data against the time to write/read the initial data. Each bar represents the sum of compression/decompression time, writing/reading the compressed data and writing/reading the initial data. We normalize the sum to 100\% and plot a dash line at 50\% to ease the comparison. It illustrates that the time of writing and reading initial data will be much longer than the time of writing and reading compressed data plus the time of compression and decompression on the Blues when the number of processors is 32 or more. This demonstrates our compressor can effectively reduce the total I/O time when dealing with the ATM data sets. We also note that the relative time spent in I/O will increase with the number of processors, because of inevitable bottleneck of the bandwidth when writing/reading data simultaneously by many processes. By contract, our compression/decompression have linear speedup with the number of processors, which means the performance gains should be greater with increasing scale.

\section{Related Work}
\label{sec: relate}

Scientific data compression algorithms fall into two categories: losseless compression \cite{fpzip,gzip,lz77} and lossy compression \cite{di,sasaki,lindstrom,laksh}.

Popular lossless compression algorithms include GZIP \cite{gzip}, LZ77 \cite{lz77}, and FPZIP \cite{fpzip}. However, the mainly limitation of the lossless compressors is their fairly low compression factor (up to 2:1 in general \cite{ratana}). In order to improve the compression factor, several lossy data compression algorithms were proposed in recent years. ISABELA \cite{laksh} performs data compression by B-spline interpolation after sorting the data series. But ISABELA has to use extra storage to record the original index for each data point because of the loss of the location information in the data series; thus, it suffers from a low compression factor especially for large numbers	 of data points. Lossy compressors using vector quantization, such as NUMARCK \cite{chen} and SSEM \cite{sasaki}, cannot guarantee the compression error within the bound and have a limitation of the compression factor, as demonstrated in \cite{di}. The difference between NUMARCK and SSEM is that NUMARCK uses vector quantization on the differences between adjacent two iterations for each data, whereas SSEM uses vector quantization on the high frequency data after wavelet transform. ZFP is a lossy compressor using exponent/fixed-point alignment, orthogonal block transform, bit-plane encoding. However, it might not respect the error bound when the data value range is huge.
% and it has a stronger correlation on the variables with relatively low compression factors in the ATM data sets.
\section{Conclusion and Future Work}
\label{sec:conclude}

In this paper, we propose a novel error-controlled lossy compression algorithm. We evaluate our compression algorithm by using multiple real-world production scientific data sets across multiple domains, and we compare it with five state-of-the-art compressors based on a series of metrics. We have implemented and released our compressor under a BSD license. The key contributions are listed below.
\begin{itemize}
\item We derive a generic model for the multidimensional prediction and optimize the number of data points used in the prediction to achieve significant improvement in the prediction hitting rate.
\item We design an adaptive error-controlled quantization and variable-length encoding model (AEQVE) to deal effectively with the irregular data with spiky changes.
\item Our average compression factor is more than 2x compared with the second-best compressor with reasonable error bounds and our average compression error has more than 3.8x reduction over the second-best with user-desired bit-rates on the ATM, APS and hurricane data sets.
\end{itemize}

We encourage users to evaluate our lossy compressor and compare with existing state-of-the-art compressors on more scientific data sets. In the future work, we plan to optimize our compression for different architectures and data sets. We will also further improve the autocorrelation of our compression on the data sets with relatively high compression factors.
%Although our compressor (including compression and decompression) is more than 2x faster than SZ and 30x faster than ISABELA and our compressor's decompression is 41\% faster than ZFP's on the average, our compressor's compression is still about 50\% slower than ZFP's. We therefore plan to optimize our compression for different architectures and high dimensional data sets in the future. Also we will further improve the autocorrelation of our compression on the data sets with relatively high compression factors in the future.
\section*{Acknowledgments}
\scriptsize
This research was supported by the Exascale Computing Project (ECP), Project Number: 17-SC-20-SC, a collaborative effort of two DOE organizations - the Office of Science and the National Nuclear Security Administration, responsible for the planning and preparation of a capable exascale ecosystem, including software, applications, hardware, advanced system engineering and early testbed platforms, to support the nation's exascale computing imperative. The submitted manuscript has been created by UChicago Argonne, LLC, Operator of Argonne National Laboratory (Argonne). Argonne, a U.S. Department of Energy Office of Science laboratory, is operated under Contract No. DE-AC02-06CH11357.



%\bibliographystyle{IEEEtran}

%\bibliographystyle{IEEEtran}
\bibliographystyle{abbrv}
\bibliography{bib/refs}

\end{document}